\documentclass[showpacs,amsmath,amssymb,twocolumn,aps,pra,superscriptaddress,notitlepage,10pt]{revtex4-2}
\usepackage{amsfonts,amssymb,amssymb}
\usepackage{amsthm}
\usepackage{amsmath}
\usepackage{algorithm}
\usepackage{algorithmicx}
\usepackage{algpseudocode}
\usepackage{enumerate}
\usepackage{xcolor}
\usepackage{cancel}
\usepackage[colorlinks, linkcolor=red, anchorcolor=blue, citecolor=purple, hypertexnames=false]{hyperref}
\usepackage[capitalise]{cleveref}
\usepackage{graphicx}
\usepackage[caption=false, justification=centerlast, label font=bf]{subfig}
\usepackage{physics}
\usepackage[most]{tcolorbox}
\usepackage[]{qcircuit}
\usepackage{appendix}
\usepackage{makecell}
\usepackage{tikz}
\usetikzlibrary{shapes.geometric, arrows.meta, positioning}
\usetikzlibrary{calc}
\usepackage{mathrsfs}
\usepackage{dsfont}
\usepackage{MnSymbol}
\newcommand{\bE}{\mathbb E}
\newcommand{\cP}{\mathcal P}
\newcommand{\cO}{\mathcal O}
\newcommand{\cM}{\mathcal E_{GC}}
\newcommand{\cE}{\mathcal E_{LU}}

\newcommand{\sE}{\mathscr E}
\newcommand{\sU}{\mathscr U}
\newcommand{\kp}{\ket{\psi}}
\newcommand{\supp}{\mathrm{supp}}
\newtheorem{theorem}{Theorem}
\newtheorem{lemma}{Lemma}

\newtheorem{definition}{Definition}

\newtheorem{fact}{Fact}

\begin{document}
\title{Taming Trotter Errors with Quantum Resources}
\date{\today}
\author{Xiangran Zhang}
\affiliation{
QICI Quantum Information and Computation Initiative, School of Computing and Data Science,
The University of Hong Kong, Pokfulam Road, Hong Kong}
\author{Jue Xu}
\affiliation{
QICI Quantum Information and Computation Initiative, School of Computing and Data Science,
The University of Hong Kong, Pokfulam Road, Hong Kong}
\affiliation{
Foresight Quantum, Shanghai, China}
\author{Qi Zhao}
\email[]{zhaoqi@cs.hku.hk}
\affiliation{
QICI Quantum Information and Computation Initiative, School of Computing and Data Science,
The University of Hong Kong, Pokfulam Road, Hong Kong}
\author{You Zhou}
\email[]{you\_zhou@fudan.edu.cn}
\affiliation{Key Laboratory for Information Science of Electromagnetic Waves (Ministry of Education), Fudan University, Shanghai 200433, China}

\begin{abstract}
Quantum simulation is a cornerstone application of quantum computing, yet how fundamental quantum resources—entanglement and non-stabilizerness (``magic")—shape simulation fidelity remains an open question.
In this work, we establish a rigorous connection between these resources and the statistical behavior of algorithmic errors arising in Hamiltonian simulation based on the Trotter-Suzuki formula. By analyzing ensembles of states with fixed entanglement entropy or magic, we make two key discoveries: First, the variance of the Trotter error decreases with increasing entanglement entropy, indicating a stronger concentration of error for entangled states. Moreover, we find that the kurtosis of the error exhibits a negative linear dependence on magic, implying that states with high magic possess lighter-tailed error distributions and thus a reduced probability of large deviations.
These findings reveal a subtle phenomenon: quantum resources that obstruct classical emulation may, paradoxically, enhance the intrinsic robustness of quantum simulation, highlighting a constructive interplay between complexity and stability in quantum computation.
\end{abstract}

\maketitle
Simulating physical systems  \cite{feynmanSimulatingPhysicsComputers1982,feynmanQuantumMechanicalComputers1985} stands as a foundational task for which quantum computers promise a transformative advantage, from quantum materials and chemical processes to high-energy physics and cosmology, offering unprecedented precision and insight \cite{georgescuQuantumSimulation2014,altmanQuantumSimulatorsArchitectures2021}.
The foundational promise of quantum simulation stems from its capacity to efficiently simulate quantum systems that are classically intractable \cite{feynmanSimulatingPhysicsComputers1982,georgescuQuantumSimulation2014}. Such classical intractability originates from the exponential scaling of computational resources required to represent essential quantum features—particularly entanglement \cite{SCHOLLWOCK201196,ciracMatrixProductStates2021,PhysRevLett.91.147902,tensornetworkandentanglement2019} and non-stabilizerness, or called magic \cite{Kitaev2004computation,Veitch2014MagicResource,guPseudomagicQuantumStates2024,guMagicInducedComputationalSeparation2025}
On the entanglement side, the bond dimension in tensor-network-like classcial simulation methods grows exponentially with the entanglement entropy of the target state \cite{PhysRevLett.91.147902,ciracMatrixProductStates2021,tensornetworkandentanglement2019}. On the magic side, while Clifford circuits can produce volume law of entanglement, and are also efficiently simulable classically due to the Gottesman-Knill theorem, the introduction of non-Clifford gates—which generate magic renders the simulation exponentially hard \cite{gottesman1998heisenbergrepresentationquantumcomputers,PhysRevA.70.052328,bravyiSimulationQuantumCircuits2019,PhysRevLett.118.090501}.
Despite the well-established role of these resources as bottlenecks for classical simulability, their influence on the performance and error characteristics of quantum simulation algorithms themselves remains largely unexplored.

\begin{figure}[tb!]
    \centering
    \includegraphics[width=\linewidth,height=.56\linewidth]{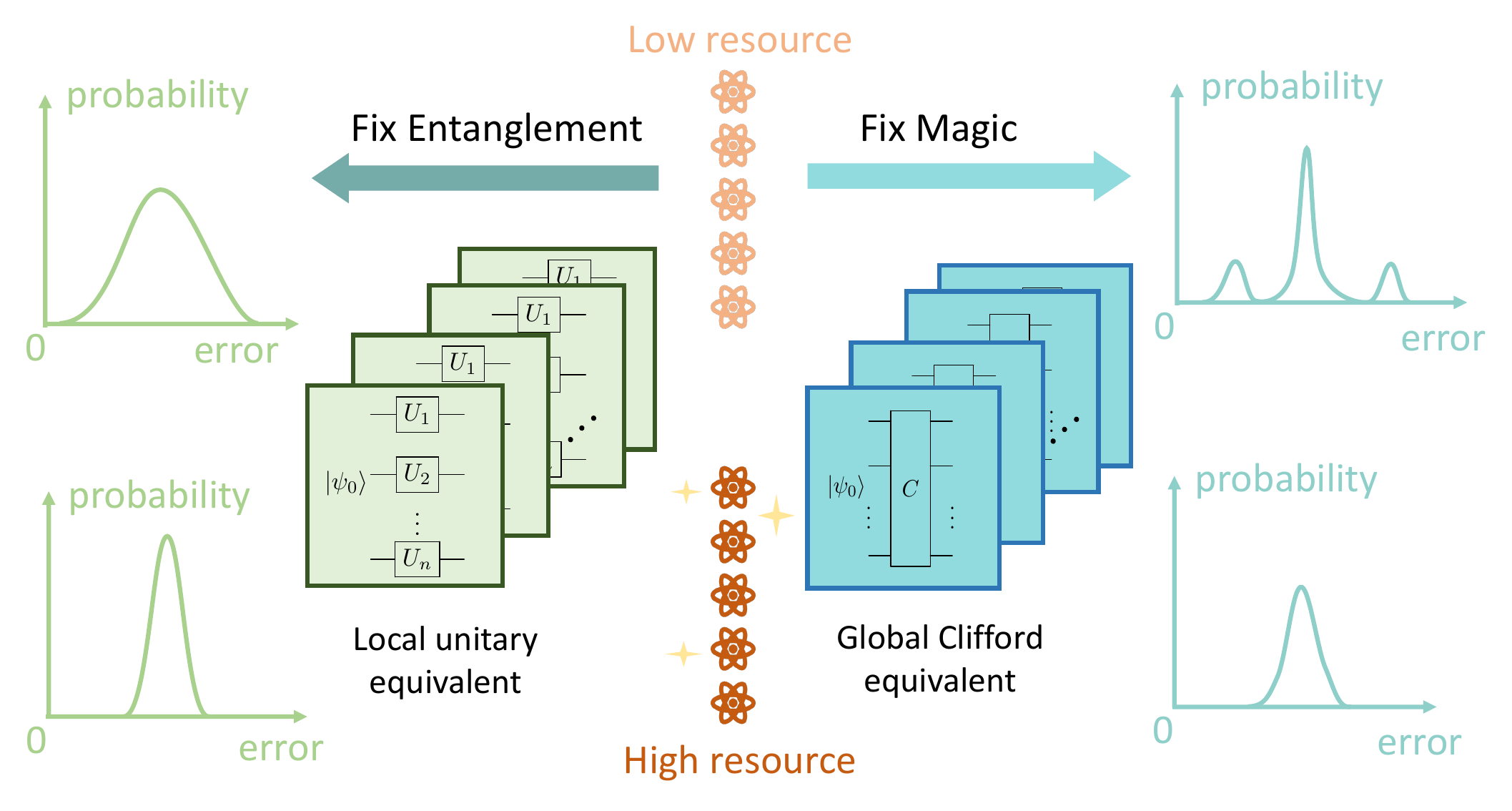}
    \caption{Entanglement reduces error size, while magic suppresses outliers. We use local-unitary-equivalent states (left) and global-Clifford-equivalent states (right) to fix entanglement and magic, respectively. Entanglement reduces the variance of error and lowers its upper bound. Meanwhile, high magic suppresses the tail of the error distribution, reducing the probability of errors that deviate significantly from the mean value.}
    \label{fig:intro}
\end{figure}

Recent research has increasingly recognized the crucial role of initial states in determining quantum simulation errors, such as low-energy states \cite{csahinouglu2021hamiltonian,Mizuta2025Trotterization,Gong2024complexityofdigital,hejazi2024betterboundslowenergyproduct,yi_spectral_2022}, random input states \cite{zhao_hamiltonian_2022,chen2024average}, and so on \cite{PhysRevA.107.L040201,PhysRevResearch.6.043155,becker2025convergence,Mobus2024strongboundstrotter,fang2025trottererrormanybodyquantum,hatomura2023firstordertrotterdecompositiondynamicalinvariant,Hatomura2022StateDependentPRA,an2021time,su2021nearly,xuExponentiallyDecayingQuantum2025}, moving beyond analyses based solely on the Hamiltonian.
Previous studies, primarily focused on entanglement  \cite{zhao2024entanglementacceleratesquantumsimulation}, further revealed that quantum resources yield little average advantage—random input states perform nearly as well as highly entangled ones in Trotterized evolution \cite{zhao2024entanglementacceleratesquantumsimulation,zhao_hamiltonian_2022,chen2024average}.
This raises a natural and fundamental question: can more general forms of quantum resource influence how simulation errors fluctuate, rather than merely their mean? Indeed, the statistical landscape of these errors—capturing their variance, concentration, and rare-event tails—remains virtually unexplored. Uncovering these effects could open new avenues toward more efficient and noise-resilient Hamiltonian simulation.

In this Letter, we establish a rigorous connection between key quantum resources—entanglement and magic—and the statistical behavior of algorithmic errors in digital quantum simulation, as illustrated in Fig.~\ref{fig:intro}.
These resources, which render quantum systems hard to emulate classically, are shown here to enhance the stability of quantum simulation itself: by analyzing ensembles of states with fixed entanglement (local-unitary-equivalent) and fixed magic (global-Clifford-equivalent), we uncover two clear trends—greater entanglement suppresses the variance of Trotter error, while higher magic reduces its kurtosis, yielding lighter-tailed distributions and fewer extreme deviations.
These results reveal an intrinsic link between quantum resources and algorithmic stability, establishing a resource-aware framework for error suppression and offering concrete guidance toward realizing practical quantum advantage.

\vspace{0.2cm}
\begin{figure}[t]
    \centering
    \includegraphics[width=\linewidth]{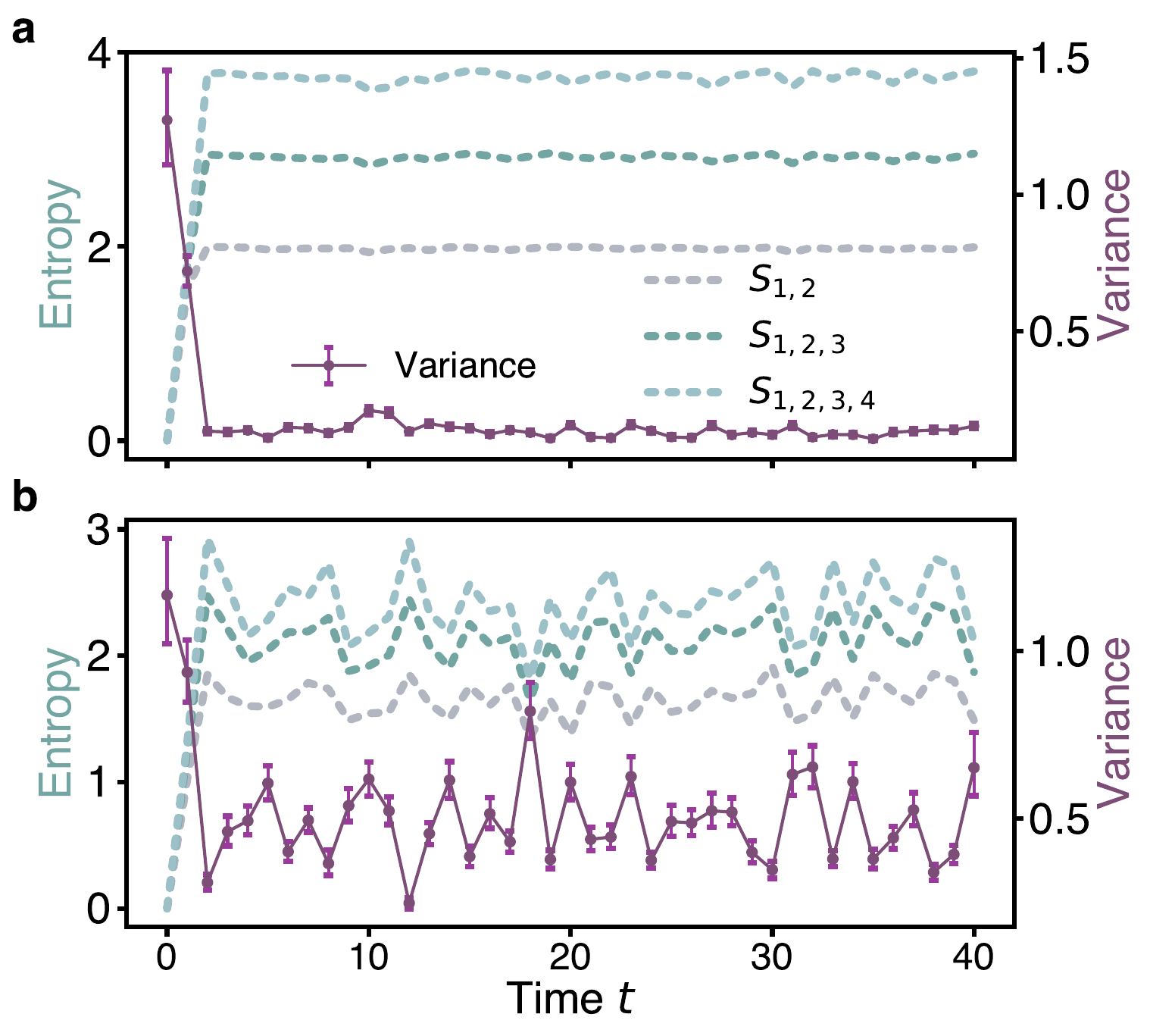}
    \caption{Entanglement reduces the variance of simulation errors. We simulate a 10-qubit system under the QIMF model with two parameter sets: one, $(h_x,h_y,J)=(0.8090,0.9045,1)$, causing rapid entanglement growth (a); and another, $(h_x,h_y,J)=(0,0.9045,1)$, with slow growth (b) \cite{kim2014testing}. The entanglement entropy of the state $\ket{\psi_t}=e^{\-iHt}\ket{\psi_0}$ (of the reduced density matrices on the first few qubits) is shown as the dashed lines (right Y-axis). For each time $t$, we take the state $\ket{\psi_t}$ and calculate the variance of simulation error using the 1st-oder Trotter method over the ensemble $\cE(\psi_t)$ in Eq.~\eqref{equ: E_LU}, shown as the solid line (left Y-axis).}
    \label{fig: var_ent}
\end{figure}

\noindent\textbf{Hamiltonian simulation with Trotter formula---}
The Trotter-Suzuki decomposition, also known as the product formula \cite{trotter1959product,suzuki1991general,doi:10.1126/science.273.5278.1073}, plays a fundamental role in quantum simulation and quantum computing \cite{childs2018toward}. It provides a method to approximate the evolution of a Hamiltonian composed of non-commuting operators by a product of their individual exponentials.

Given a Hamiltonian $H$ and a short evolution time $\delta t$, we denote the ideal evolution operator as $\sU_0(\delta t)=e^{-iH\delta t}$. The first-order product formula (PF1) algorithm for a Hamiltonian with decomposition $H=\sum_{l=1}^L H_l$ applies the unitary operation $\sU_1(\delta t):=e^{-iH_1\delta t}e^{-iH_2\delta t}\cdots e^{-iH_L\delta t}=\overrightarrow\prod_le^{-iH_l\delta t}$. Here, the right arrow indicates the product is in the order of increasing indices. This provides an effective simulation method when the terms $e^{-iH_l\delta t}$ are feasible to implement. The second-order product formulas (PF2) can be obtained by combining evolutions in both increasing and decreasing orders of indices, with $\sU_2(\delta t)=\overrightarrow\prod_le^{-iH_l\delta t/2}\overleftarrow\prod_le^{-iH_l\delta t/2}$. More generally, one can construct $p$th-order product formulas $\sU_p(\delta t)$, for even $p$, recursively from the second-order formula \cite{yoshida_construction_1990,suzuki1991general,Morales_2025}.
To evaluate the effectiveness of a $p$th-order product formula (PF$p$) simulation, a natural approach is to consider the spectral norm of the difference between two operators $\sU_0(\delta t)$ and $\sU_p(\delta t)$, i.e.  $\|\sU_0(\delta t)-\sU_p(\delta t)\|$. It is important to note that here $\|\sU_0(\delta t)-\sU_p(\delta t)\|=\max _{\ket{\psi}}\|(\sU_0(\delta t)-\sU_p(\delta t)\ket{\psi}\|$ represents the worst-case scenario for all possible state $\ket{\psi}$. However, for states with some prior knowledge (like energy \cite{Gong2024complexityofdigital,csahinouglu2021hamiltonian,Mizuta2025Trotterization} and entanglement \cite{zhao2024entanglementacceleratesquantumsimulation}), the Trotter error could be smaller and does not necessarily reach the worst case.
Therefore, it is more reasonable to consider the state-specific error $\|(\sU_0(\delta t)-\sU_p(\delta t)\ket{\psi}\|$ for given initial states $\ket{\psi}$.

The dominant term of the Trotter error can be obtained by the decomposition according to the commutator scaling method \cite{childs2021theory}, i.e.  $\sU_0(\delta t)-\sU_p(\delta t)=E\delta t^{p+1}+\sE_{re}$, where $\|\sE_{re}\|=\cO(\delta t^{p+2})$. Here $E=\sum_j E_j$, is the total leading-order error with local terms $E_j$ for local Hamiltonians. Consider a Hamiltonian $H=A+B$ with PF1, the dominant term $E$ is determined by the commutators of $A$ and $B$, say, $E=[A,B]$. For example, for the one-dimensional quantum Ising spin model with mixed fields (QIMF) on $N$-qubit, the Hamiltonian shows $H=h_x\sum_{j=1}^N X_j+h_y\sum_{j=1}^NY_j+J\sum_{j=1}^{N-1}X_jX_{j+1}$, with $X_j,Y_j$ denotes the Pauli operators on qubit $j$.
Setting $A$ consists of all the $X$ terms and $B$ consists of all the $Y$ terms, $\sU_0(\delta t)$ could be simulated by $\sU_1(\delta t)=e^{-iA\delta t}e^{-iB\delta t}$. Then the leading term of its error will be local, i.e.  of the form $E=\sum E_j$ with each error term
\begin{equation}
    E_j=2ih_xh_yZ_j+2iJh_y(Z_jX_{j+1}+X_jZ_{j+1})\label{equ: Ej}
\end{equation}
acting on only two (adjacent) qubits. Similarly, in the vast majority of cases, the error terms in $E$ exhibit locality.

To analyze state-specific errors, the following discussion focuses on the squared magnitude of the leading error operator \( E \) for a given initial state \( \ket{\psi} \):
\begin{equation}
    s_E(\psi):=\|E\ket{\psi}\|^2=\bra{\psi}E^\dagger E\ket{\psi}.
    \label{equ:definition}
\end{equation}
As the true error satisfies $\norm{(\sU_0(\delta t)-\sU_p(\delta t))\ket{\psi}}^2=s_E(\psi)\delta t^{2p+2}+\mathcal O(\delta t^{2p+4})$, the quantity $s_E(\psi)$ serves as a convenient measure for analytical and numerical treatment when $\delta t$ is small, and more importantly, reflects how the quantum resources encoded in \( \ket{\psi} \) influence the simulation error.

Before showing the main results, we first define some notations. We focus on $N$-qubit system, with $d=2^N$ representing the total dimension. $\cP_N$ represents the
$N$-qubit Pauli group, quotient by the global phase, i.e.  $\cP_N=\{\bigotimes_{i=1}^N P_i, P_i\in\{I_2,X,Y,Z\}\}$. The support of an operator $A$, denoted by $\supp(A)$, is the set of qubits on which it acts nontrivially.
$\|A\|_F :=\sqrt{\Tr(AA^\dagger)/d}$ denotes the normalized Frobenius norm of the operator $A$ with dimension $d$.

\vspace{0.2cm}
\noindent\textbf{Error Variance and Quantum Entanglement---}
We first discuss the impact of entanglement on quantum simulation, and are concerned with the statistical properties of the Trotter error exhibited by all states possessing a fixed amount of entanglement.

As such, we define the random state ensemble

\begin{equation}
    \cE(\psi_0):=\left\{\left(\mathop{\bigotimes}_{i=1}^NU_i\right) \ket{\psi_0}\right\}\label{equ: E_LU}
\end{equation}
where each $U_i$ is a single-qubit Haar-random unitary \cite{Haarmeasure,Mele_2024,harrow2013church} (actually $U_i$ needs only to be a unitary $2$-design in our proof). It is clear that all states in $\cE(\psi_0)$ are local-unitary-equivalent to $\ket{\psi_0}$, thus own the same entanglement. In this way, we can bound the variance of the simulation error in the following theorem.

\begin{theorem}\label{theo: variance}
    For an initial state $\ket{\psi_0}$ and the leading error $E=\sum E_j$ of Hamitonian simulation algorithm, the variance of the state error term $s_E(\psi)$ given in Eq.~\eqref{equ:definition}, evaluated under the distribution $\cE(\psi_0)$ defined in Eq.~\eqref{equ: E_LU} is bounded by the entanglement entropy of subsystems,
    \begin{align}\label{eq:ajj}
        \mathop{\mathrm{Var}}_{\ket{\psi}\sim\cE(\psi_0)}[s_E(\psi)] \leq
        \sum_{jj'}a_{jj'}\sqrt{2[\log(d_{jj'})-S(\rho_{jj'})]},
    \end{align}
    where $S(\rho)=-\Tr[\rho\log(\rho)]$, and the reduced density matrix $\rho_{jj'}=\Tr_{[N]\setminus\supp(jj')}(\ketbra{\psi_0}{\psi_0})$ and $d_{jj'}=2^{\abs{\supp(jj')}}$, with $\supp(jj'):=\supp(E_{j'}^\dagger E_j)$ for short.
\end{theorem}

See Section \uppercase\expandafter{\romannumeral2} in Supplemental Material \cite{seesm} for the detailed proof.  Theorem \ref{theo: variance} shows that for a set of states with higher entanglement, the variance of the error should be smaller than that for states with lower entanglement, as numerically certified in FIG.~\ref{fig: var_ent} with QIMF models. In addition, the local unitary introduced in \cref{equ: E_LU} forms an 1-design, so $\bE_{\ket{\psi}\sim\cE(\psi_0)} s_E(\psi)=\|E||_F^2$ for any $\ket{\psi_0}$ \cite{zhao_hamiltonian_2022}, i.e., sharing the same mean value. This suggests that the Trotter errors of states with higher entanglement tend to cluster more closely around $\|E\|_F^2$. One can readily observe this by applying probabilistic inequalities such as Chebyshev's inequality, $\mathop{\mathrm{Pr}}_{\ket{\psi}\sim\cE(\psi_0)}\left[\left|s_E(\psi)-\|E\|_F^2\right|\ge k\sigma\right]\le\frac{1}{k^2}$, where $\sigma=\sqrt{\mathrm{Var}(s_E(\psi))}$ denotes the standard deviation.

\vspace{0.2cm}
\begin{figure}[t]
	\centering
    \includegraphics[width=\linewidth]{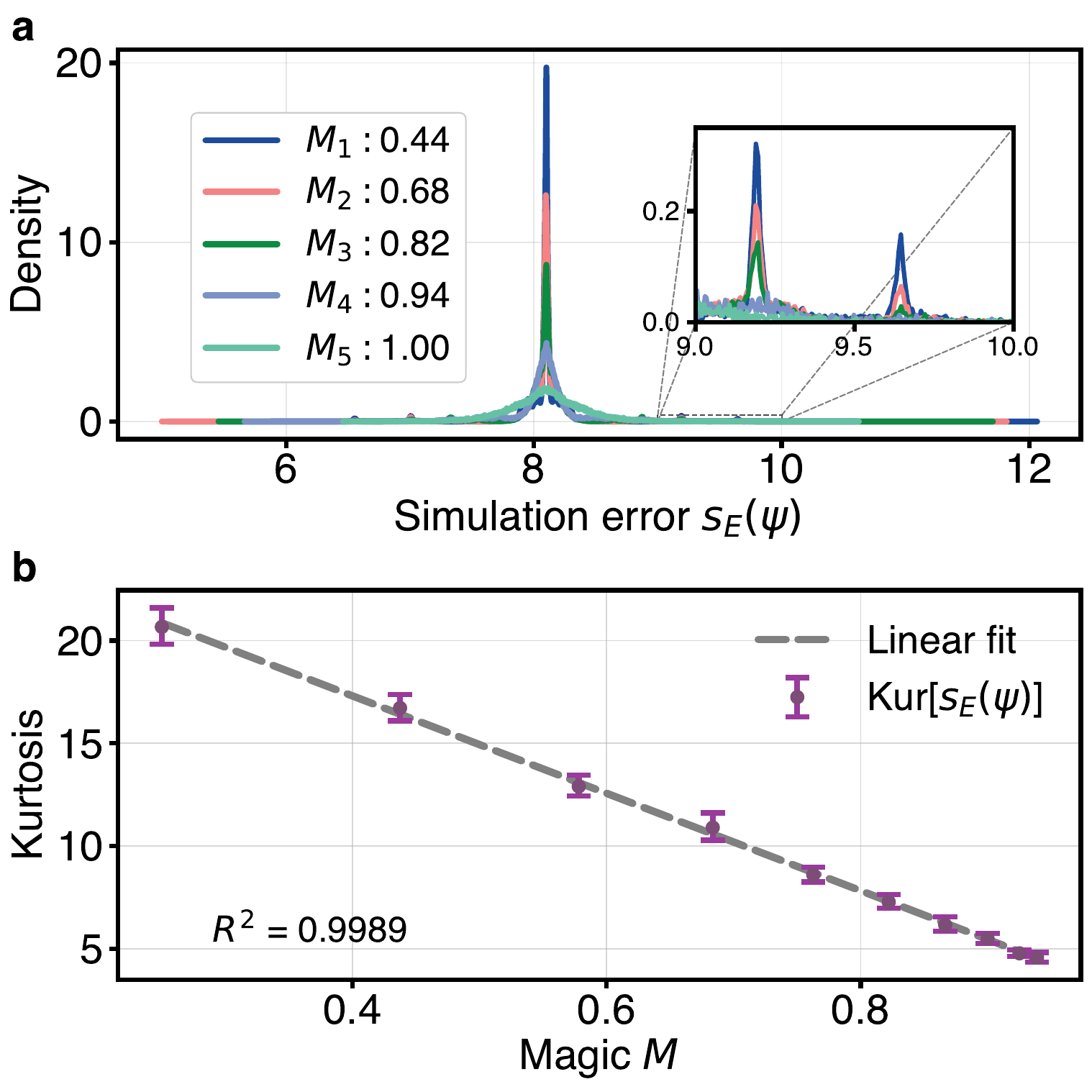}
    \caption{
    Magic reduces the kurtosis of simulation errors of QIMF with $(h_x,h_y,J)=(0.8090,0.9045,1)$.
    (a) Error distributions for the ensemble in Eq.~\eqref{eq:gClen} using five starting states $\ket{\psi_k}$ with different magic values $M_k$. The inset clearly shows a higher proportion of outliers in the low-magic ensembles.
    (b) The corresponding relationship between kurtosis and magic, exhibiting a clear negative linear correlation, as predicted by Eq.~\eqref{eq:MnegK}.
    }
    \label{fig: kur_mag}
\end{figure}

\noindent\textbf{Error Kurtosis and Quantum Magic---} In magic-state resource theory, Clifford operations and stabilizer states are considered ``free" as a direct consequence of the Gottesman-Knill theorem \cite{gottesman1998heisenbergrepresentationquantumcomputers,PhysRevA.70.052328}, which establishes their efficient classical simulability. In contrast, non-Clifford gates—though essential for achieving universality—cannot be efficiently simulated and are therefore considered resource-intensive. Furthermore, for some quantum error correction codes, while Clifford operations can be implemented fault-tolerantly \cite{PhysRevLett.104.030503}, non-Clifford gates pose significant challenges in fault-tolerant constructions \cite{Kitaev2004computation,road2017,PhysRevLett.102.110502,bravyiMagicstateDistillationLow2012,PhysRevLett.118.090501,niroulaPhaseTransitionMagic2024,willsConstantoverheadMagicState2025}. Additionally, non-stabilizerness renders quantum computation exceptionally difficult to verify \cite{Jose2024GainingConfidencePRR}. These fundamental distinctions justify the treatment of non-stabilizerness or called magic as a valuable resource and a meaningful metric in quantum information processing. Recently, magic has also emerged as a focal point of heightened research activity within realms such as many-body physics \cite{liuManyBodyQuantumMagic2022,Chen2024magicofquantum,PRXQuantum.4.040317,turkeshiPauliSpectrumNonstabilizerness2025} and quantum chaos \cite{leoneQuantumChaosQuantum2021,Passarelli2025chaosmagicin,PhysRevD.106.126009}.

There are various measures of non-stabilizerness,
such as $\alpha$-R\'enyi stabilizer entropies and linear $\alpha$-stabilizer entropies \cite{leoneStabilizerRenyiEntropy2022,Haug_2023,leone2024stabilizer,dowlingMagicResourcesHeisenberg2025,lamiNonstabilizernessPerfectPauli2023,tarabungaNonstabilizernessMatrixProduct2024,li2025invested}. Given a pure state $\ket{\psi}$, the quantity $d^{-1}|\bra{\psi}P\ket{\psi}|^2$ forms a probability distribution on the Pauli group $\cP_N$. Define the $\alpha$-stabilizer purity $P_\alpha:=\frac{1}{d}\sum_{P\in \cP_N} |\bra{\psi}P\ket{\psi}|^{2\alpha}$, and then linear $\alpha$-stabilizer entropy shows $M^{\mathrm{lin}}_\alpha(\psi):=1-P_\alpha$. In this Letter, we choose $M(\psi):=M_2^{\mathrm{lin}}(\psi)$ with $\alpha=2$ to quantify the magic of the state $\ket{\psi}$, which is a strong stabilizer monotone \cite{leone2024stabilizer} and also feasible to compute.

Now we investigate the impact of magic on quantum simulation.
By definition, one has the fact that $M(\ket{\psi})=M(C\ket{\psi})$ for all Clifford unitary $C$.
We thus choose the random ensemble $\cM(\psi_0)$ with each state inside obtained by applying a random Clifford gate $C\in\mathrm{Cl}(2^N )$ to a starting state $\ket{\psi_0}$,

\begin{equation}\label{eq:gClen}
    \cM(\psi_0):=\left\{C\ket{\psi_0}\right\}.
\end{equation}

The Clifford group forms a 3-design but not a 4-design \cite{webb2015clifford,zhu2016clifford,bittel2025completetheorycliffordcommutant}, so the first three moments of $s_E(\psi)$ equals the Haar random one, and thus will not exhibit any dependence on different initial state $\ket{\psi_0}$. As a result, it is natural to consider the fourth moment, and we investigate the kurtosis of the error term here. The definition of the kurtosis of $s_E(\psi)$ shows \cite{westfall2014kurtosis},
\begin{equation}\label{eq:kurtosisSE}
    \mathop{\mathrm{Kur}}_{\ket{\psi}\sim\cM(\psi_0)}(s_E(\psi))=\bE\left[\left(s_E(\psi)-\|E\|_F^2\right)^4\right]/\sigma^4,
\end{equation}
with $\sigma$ the standard deviation of $s_E(\psi)$. As the average of the fourth power of the pivot, i.e., $(X-EX)/\sigma$, kurtosis primarily derives its contributions from extreme values. Intuitively, higher kurtosis indicates a relatively greater proportion of outliers in the data.

By applying the properties of the 4th-moment function of Clifford group \cite{zhu2016clifford,bittel2025completetheorycliffordcommutant,Zhou2023performanceanalysis}, we can bound the kurtosis of the leading-order error term via the following Theorem \ref{theo: kurtosis}.

\begin{theorem}\label{theo: kurtosis}
    For an initial state $\ket{\psi_0}$ with magic $M(\psi_0)$ and leading error $E=\sum E_j$ of $p$th-order product formula, the kurtosis of the error term $s_E(\psi)$ given in Eq.~\eqref{eq:kurtosisSE} under the distribution $\cM(\psi_0)$ defined in Eq.~\eqref{eq:gClen} is linear with the magic $M(\psi_0)$. Specifically, the kurtosis satisfies
    \begin{equation}\label{eq:MnegK}
        \mathop{\mathrm{Kur}}_{\ket{\psi}\sim\cM(\psi_0)}\left[s_E(\psi)\right]=\alpha +\beta M(\psi_0),
    \end{equation}
    where $\alpha,\beta$ are some coefficients only related to $E$. If $E$ consists of at most $\mathrm{poly}(N)$ terms, the dominant term of $\beta$ is negative.
\end{theorem}

See Section \uppercase\expandafter{\romannumeral3} in Supplemental Material \cite{seesm} for the proof and the detailed expressions of $\alpha$ and $\beta$.

Here we give some discuss about the condition that $E$ consists of $\mathrm{poly}(N)$ terms. For a Hamiltonian $H=A+B$ consisting of $\mathcal O(N^k)$ terms, the error matrix $E$ of $p$-th order product formula has at most $\mathcal O(N^{pk})$ terms (usually much less than this by locality), which can be easily obtained from the linearity of the commutator.
Standard models in quantum simulation (such as the Ising, Heisenberg, and Fermi-Hubbard models) comprise only $\mathcal O(N^k)$ local terms. Therefore, the error terms of constant-order product formulas satisfy this requirement for the majority of Hamiltonian simulation scenarios.
Furthermore, the absolute value of the leading term in $\beta$ increases exponentially faster than the other terms (see the SM \cite{seesm}). Therefore, $\beta$ turns negative even when $N$ is not very large. For instance, FIG.~\ref{fig: kur_mag}(b) displays the case of $N=10$, where $\beta$ is already negative.

A smaller kurtosis generally implies thinner tails, which means a lower probability of extreme values occurring \cite{balanda1988kurtosis,he2010bounding}. For an intuitive understanding, we examine the following inequality given in Ref.~\cite{zelen1954bounds}: define $\kappa_m:=\bE[(s_E(\psi)-\|E\|_F)^m]/\sigma^m$,
with $\kappa_4$ being the kurtosis. Then for $t>(\kappa_3+\sqrt{\kappa_3^2+4})/2$,
$$\mathop{\mathrm{Pr}}_{\ket{\psi}\sim \cM}\left[s_E(\psi)-\|E\|_F^2\ge t\sigma\right]\le\left(1+t^2+\frac{(t^2-t\kappa_3-1)^2}{\kappa_4-\kappa_3^2-1}\right)^{-1}.$$
For any distribution function, it always holds that $\kappa_4\ge\kappa_3^2+1$ \cite{pearson1916ix}. Therefore, when two distributions share the same $\kappa_3$, a smaller $\kappa_4$ indicates a smaller probability of $s_E(\psi)\gg\norm{E}_F^2$ happening, which is also clearly certified by the numerical results in FIG.~\ref{fig: kur_mag}: a state with greater magic implies a smaller probability that its Trotter error deviates significantly from the mean value.

\begin{figure}[t!]
    \centering
    \includegraphics[width=\linewidth]{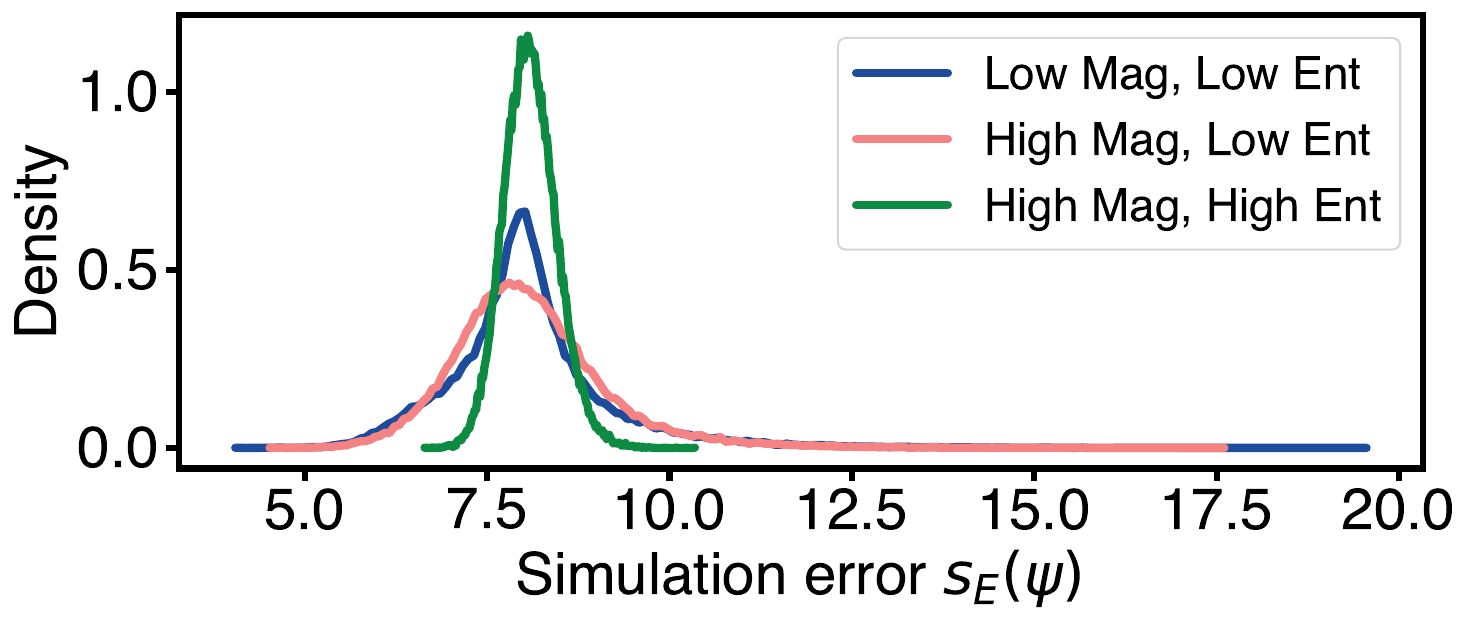}
    \caption{Error distribution of $\mathcal E_{LC}(\psi_{LL})$, $\mathcal E_{LC}(\psi_{LH})$ and $\mathcal E_{LC}(\psi_{HH})$. We simulate the corresponding ensemble $\mathcal E_{LC}$ by applying local random Clifford gates $10^6$ times to each of these three initial states chosen from FIG.~\ref{fig: increase}.
    }
    \label{fig: threeDistribution}
\end{figure}

\vspace{0.2cm}
\noindent\textbf{Joint effects of entanglement and magic---} To further explore the joint influence of entanglement and magic on error statistics of quantum simulation, we define the following ensemble
\begin{equation}\label{eq:esembleLC}
    \mathcal E_{LC}(\psi_0):=\left\{\left(\bigotimes_{i=1}^N C_i\right)\ket{\psi_0}\right\},
\end{equation}
where each $C_i$ is a single-qubit random Clifford gate on $i$-th qubit. This construction preserves both the entanglement and magic of the initial state $\ket{\psi_0}$, while randomizing local bases through local Clifford operations. Since the Clifford group forms a unitary 3-design \cite{webb2015clifford,zhu2016clifford,bittel2025completetheorycliffordcommutant}, the variance bound derived in Theorem \ref{theo: variance} remains applicable to $\mathcal E_{LC}$ here.

Though a full analytical treatment of the resulting kurtosis for the local Clifford ensemble becomes challenging due to the interplay between entanglement and magic, we turn to numerical investigations using states evolved under the QIMF model with different parameters. In FIG.~\ref{fig: increase} of the End Matter, we show how the entanglement and magic increase under the physical evolution.

We select three representative states from FIG.~\ref{fig: increase}: $\ket{\psi_{LH}}$ (low entanglement, high magic), $\ket{\psi_{LL}}$, and $\ket{\psi_{HH}}$. Taking these as initial states, FIG.~\ref{fig: threeDistribution} displays the corresponding error distributions generated by the ensemble $\mathcal E_{LC}$ in \cref{eq:esembleLC}. The results reveal a clear and consistent trend: magic effectively suppresses extreme outliers, while entanglement enhances the concentration of errors. These observations align well with our earlier analytical insights, reinforcing the complementary roles of magic and entanglement in shaping error statistics.

\vspace{0.2cm}
\noindent\textbf{Long-time behavior---}
The previous discussion mainly focused on the one-step Trotter error. For the long-time magic--kurtosis relation, one may instead use the accumulated error operator $(E_{\mathrm{LT}}:=\sU_p^r-\sU_0^r)$, for which the formal linear dependence remains, although the sign of the coefficient is no longer fixed by our present proof; numerically, we still observe a negative slope for the systems considered here, as shown in the Supplemental Material \cite{seesm}. We next turn to the long-time behavior associated with entanglement.

For large time $t$, the full evolution interval can be uniformly partitioned into $r$ steps of duration $\delta t=t/r$.
The true evolution $\sU_0(t)=\sU_0(\delta t)^r$ is commonly approximated by the product $\sU_p(\delta t)^r$ of the approximate unitary $\sU_p(\delta t)$, abbreviated as $\sU_p$ for brevity. The long-time Trotter error can be bounded via the triangle inequality \cite{seesm}.

\begin{equation}
    \begin{aligned}
        \norm{(\sU_p^r-\sU_0^r)\ket{\psi_0}}\le\sum_{k=0}^{r-1}\norm{(\sU_p-\sU_0)\sU_p^k\ket{\psi_0}}.\label{equ:tri-ineq}
    \end{aligned}
\end{equation}

For simplicity, we denote $\sum_{k=0}^{r-1}\norm{(\sU_p-\sU_0)\sU_p^k\ket{\psi_0}}$ as $e_r$. We define $S_k=\{s\subseteq[N]|s=\operatorname{supp}(\sU_p^{\dagger k}E_j^\dagger E_{j'}\sU_p^k),\forall j,j'\}$ as the set of supports for each term in the expansion of operator $\sU_p^{\dagger k}E^\dagger E\sU_p^k$. Subsequently, we establish the following theorem.

\begin{theorem}
    Given a state $\ket{\psi_0}$ with entanglement structure $S(\rho_A)$, neglecting the influence of higher-order terms in $\delta t$, the variance of the $p$-th order Trotter error under $\cE(\psi_0)$ satisfies
    \begin{equation}
        \mathop{\operatorname{Var}}_{\kp\sim\cE(\psi_0)}[e_r]\le r\delta t^{2p+2}\sum_{k=0}^{r-1}\sqrt{\sum_{s\in S_k}c_s\sqrt{\log d_s-S(\rho_s)}},
    \end{equation}
    where $c_s$ is a constant that is independent of the state $\psi_0$.
\end{theorem}

See Section \uppercase\expandafter{\romannumeral4} in Supplemental Material \cite{seesm} for detailed proof. The numerical results are illustrated in FIG.~\ref{fig:0.1longtime} in the End Matter. When $\delta t$ is negligibly small, the long-time behavior is similar to the one-step behavior.

\vspace{0.2cm}
\noindent\textbf{Discussions and outlook---} In this Letter, we establish a fundamental link between intrinsic quantum resources—entanglement and magic—and the statistical behavior of algorithmic error in digital quantum simulation. Entanglement suppresses the variance of Trotter error, while magic reduces its kurtosis, jointly stabilizing quantum dynamics.
Thus, resources that obstruct classical simulation can, paradoxically, enhance quantum reliability.
As entanglement and magic typically proliferate under generic dynamics \cite{kim2014testing,toniolo_dynamical_2025,Turkeshi_2025} (also demonstrated in FIG.~\ref{fig: increase}), most states are inherently biased toward favorable error statistics, rendering worst-case bounds statistically irrelevant. This intrinsic robustness suggests that realistic quantum simulations may be far more stable than pessimistic analyses imply.
Looking ahead, extending this resource-based statistical framework to other algorithms and noisy devices could reveal how quantum resources not only empower computation but also protect it—hinting at a deeper unity between complexity, accuracy, and reliability in quantum dynamics.

\vspace{0.2cm}
\textit{Acknowledgements.}--
X. Z., J. X., and Q. Z. acknowledge funding from Quantum Science and Technology-National Science and Technology Major Project No. 2024ZD0301900, National Natural Science Foundation of China (NSFC) via Projects No. 12347104 and No. 12305030, Hong Kong Research Grant Council (RGC) via No. 27300823, No. 17310926, No. N\_HKU718/23, and No. R6010-23. Y.Z. acknowledges funding  from the National Natural Science Foundation of China (NSFC) Grant No.~ 12575012 and 12205048, the Quantum Science and Technology-National Science and Technology Major Project Grant Nos.~2024ZD0301900 and 2021ZD0302000, the Shanghai QiYuan Innovation Foundation, the Shanghai Municipal Commission of Science and Technology with Grant No.~25511103200, the Shanghai Science and Technology Innovation Action Plan Grant No.~24LZ1400200, the Shanghai Pilot Program for Basic Research - Fudan University 21TQ1400100 (25TQ003), the CCF-Quantum CTek Superconducting Quantum Computing CCF-QC2025006.

\textbf{Code availability}: The code used in this study is available on \href{https://github.com/Xran-ZH/Quantum-Resource-and-Quantum-Simulation}{GitHub}.


%

\clearpage

\section*{End Matter}

\subsection{Variance and Entanglement}
To show the influence of entanglement on the variance of Trotter errors, we consider the one-dimensional (1D) QIMF of $N=10$ qubits in FIG.~\ref{fig: var_ent}.

When the evolution satisfies the Eigenstate thermalization hypothesis (ETH), such as for the QIMF model with the typical parameters $(h_x, h_y, J) = (0.8090, 0.9045, 1)$\cite{kim2014testing}, the subsystem entanglement entropy increases fast (see FIG.~\ref{fig: var_ent}(a), where the dash lines show the entropies increase). At each time $t$, we have the corresponding state $\ket{\psi_t}=e^{-iHt}\ket{0}$, then we generate another 2000 states $\hat{\mathcal{E}}_{LU}(t)=\left\{\ket{\psi_t^{(p)}}\right\}$ with the same entanglement as $\ket{\psi_t}$ by acting local random unitary on $\ket{\psi_t}$, as shown in the green box in FIG. \ref{fig:intro}.
Then, we use $\hat{\mathcal E}_{LU}(t)$ to approximate $\cE(\psi_t)$. We calculate the PF1 error $e_t^{(p)}(\delta t)=\left\|\left(\sU_0(\delta t)-\sU_1(\delta t)\right)\ket{\psi_t^{(p)}}\right\|^2$ for each $k$.
Then we use $\hat{s}^{(p)}(\psi_t):=e_t^{(p)}(\delta t)/\delta t^4$ to estimate $s_E(\psi_t)$, where we set $\delta t=0.01$.
We computed the sample variance of $\{\hat{s}^{(p)}(\psi_t)\}$ for each time point $t$ and plotted it in the figure.

We obtain FIG.~\ref{fig: var_ent}(b) by adjusting the parameters to an anomalous selection $(h_x,h_y,J)=(0,0.9045,1)$, where the entanglement growth is not significant, and repeating the same procedure as in FIG.~\ref{fig: var_ent}(a).

\subsection{Kurtosis and Magic}
To show the influence of magic on the kurtosis of Trotter errors, we construct a set $\{\ket{\psi_k}\}$ consisting of 10 starting states with different magic content $M_k:=M(\ket{\psi_k})$ correspondingly. Here we construct $\ket{\psi_k}$ by using $T$-gates and Clifford gates, $\ket{\psi_k}=C_2(T^{\otimes k}\otimes I_{N-k})C_1\ket{0}^{\otimes N}$, i.e.
\begin{align}
    \ket{\psi_k}=\quad\vcenter{\hbox{\Qcircuit @C=0.35cm @R=0cm {
         &\ket{0}&&\multigate{5}{C_1} & \gate{T_1} & \qw&\multigate{5}{C_2} & \qw  \\
         &\vdots && \ghost{C_1}\qw & \vdots &  & \ghost{C_2} & \qw \\
         &\ket{0}&& \ghost{C_1} & \gate{T_k} &\qw &\ghost{C_2} & \qw\\
         &\ket{0}&& \ghost{C_1} & \qw &\qw &\ghost{C_2} & \qw\\
         &\vdots&& \ghost{C_1}\qw & \qw & \qw&\ghost{C_2} & \qw\\
         &\ket{0}&& \ghost{C_1} &\qw &\qw & \ghost{C_2}  & \qw
     }}}\quad,
\end{align}
where $N=10$, $C_1$ and $C_2$ are random Clifford gates, and then we numerically calculate their magics $M_k$.

On each state $\ket{\psi_k}$, we operate random global Clifford gates repeatedly to generate another $10^6$ states $\hat{\mathcal{E}}_{GC}(k)=\left\{\ket{\psi_k^{(p)}}\right\}$ with the same magic $M_k$. The we use $\hat{\mathcal{E}}_{GC}(k)$ to approximate $\cM(\psi_k)$. We set the Hamiltonian as the QIMF model with parameters $(h_x,h_y,J)=(0.8090,0.9045,1)$, then compute the PF1 error $e_k^{(p)}(\delta t)=\left\|\left(\sU_0(\delta t)-\sU_1(\delta t)\right)\ket{\psi_k^{(p)}}\right\|^2$ for each $p$. Then we use $\hat{s}^{(p)}(\psi_t):=e_t^{(p)}(\delta t)/\delta t^4$ to estimate $s_E(\psi_t)$, where $\delta t=0.1$ is the evolution time. FIG.~\ref{fig: kur_mag}(a) depicts the distributions of $\left\{\hat s^{(p)}\right\}$ corresponding to five different values of magic $M_k$. In FIG.~\ref{fig: kur_mag}(b), we illustrate the relationship between the kurtosis of set $\left\{\hat s^{(p)}\right\}$ and the corresponding magic $M_k$.

\begin{figure}[t]
    \centering
    \includegraphics[width=0.95\linewidth]{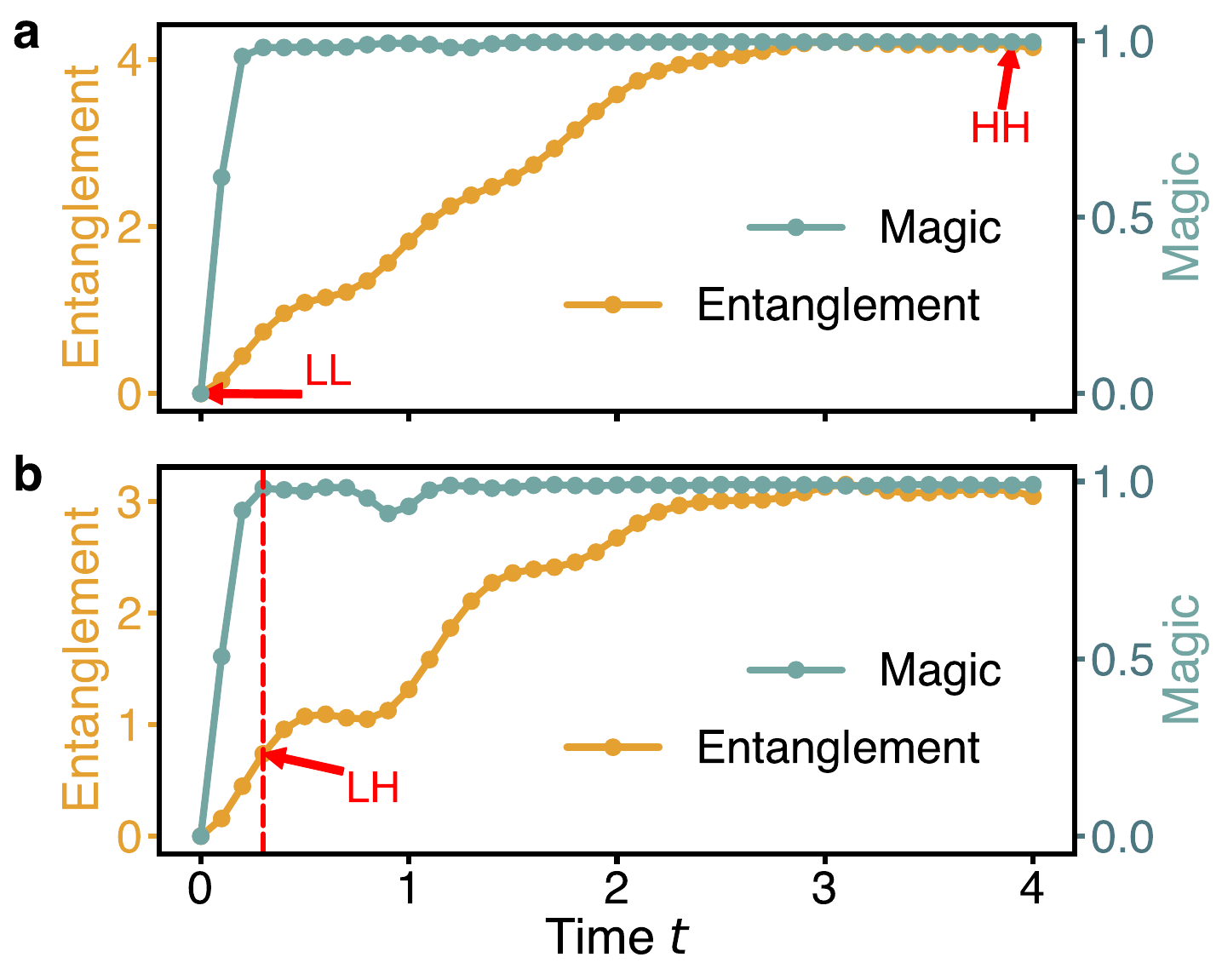}
    \caption{Both entanglement and magic exhibit growth during the time evolution.
Beginning with the initial product state $\ket{\psi_0}=\ket{0}^{\otimes N}$ for $N=10$, we evaluate how the entanglement entropy of the subsystem consisting of the first 5 qubits and magic change over time.
The results for the typical QIMF dynamics  are shown in FIG.~(a),
while FIG.~(b) displays the corresponding behavior for the atypical QIMF.}
    \label{fig: increase}
\end{figure}

\subsection{CI of Variance and Kurtosis}
We use bootstrap method to give the approximate 95\% confidence interval (CI) of variance and kurtosis in FIG.~\ref{fig: var_ent} and FIG.~\ref{fig: kur_mag}(b). The following outlines the bootstrap procedure.

First, we have an unknown cumulative distribution function (c.d.f.) $f(x)$ of random variable $X$ (here the distribution of $s_E$ under $\cE$ and $\cM$) and want to estimate a property $\phi(f)$ of $f$, such as variance and kurtosis in our main text.

Second, we generate a sample $\mathbf X=\{X_1,\dots,X_n\}, X_i\overset{idd}{\sim} f$ with sample scale $n$. We construct the empirical distribution fuction
\begin{equation}
    \hat{f}(x)=\frac{1}{n}\sum_{i=1}^n1(X_i<x).
\end{equation}
Then we construct a estimator $T(\mathbf X,f)$ of $\phi(f)$ (sample variance and sample kurtosis here).

Third, we need to evaluate the bias of $T(\mathbf X,f)$, i.e. $b(T,\phi,f)=T(\mathbf X,f)-\phi(f)$. To achieve this, we use $\hat f$ to simulate the true distribution $f$ and generate another $m$ samples with sacle $n$, $\mathbf{X}^1,\dots,\mathbf{X}^m$. For each sample $\mathbf{X}^i$, we calculate the estimator $T(\mathbf{X}^i,\hat f)$. Then we have an empirical c.d.f. of $T(\mathbf X,\hat f)-T(\mathbf X,f)$
\begin{equation}
    g(x)=\frac{1}{m}\sum_{i=1}^b1\left(T(\mathbf{X}^i,\hat f)-T(\mathbf X,f)<x\right).
\end{equation}
Last, we use $g(x)$ to simulate the behavior of $b(T,\phi,f)$. And we calculate the 2.5\% and 97.5\% quantiles $q_{0.025}$ and $q_{0.975}$ to give an approximate 95\% CI of $\phi(f)$,
\begin{equation}
    \mathrm{CI}=[T(\mathbf{X},f)-q_{0.975},T(\mathbf{X},f)-q_{0.025}].
\end{equation}

\subsection{Resources increase with quantum evolution}
FIG.~\ref{fig: increase} illustrates the growth of entanglement and magic during evolution under QIMF, with system size $N=10$ and entanglement quantified by the $5$-qubit subsystem entanglement entropy. As shown, entanglement exhibits a rapid increase under typical QIMF dynamics, while magic rises quickly and saturates, regardless of whether the QIMF is typical, i.e. $(h_x, h_y, J) = (0.8090, 0.9045, 1)$, or atypical, i.e. setting $h_x=0$.

For generic Hamiltonians, the states involved in quantum simulations generally possess substantial entanglement and magic—precisely the regime where our approach demonstrates optimal performance. In contrast, classical tensor network methods struggle with highly entangled states \cite{tensornetworkandentanglement2019}, and simulations based on the Gottesman-Knill theorem \cite{gottesman1998heisenbergrepresentationquantumcomputers} face exponential complexity in the presence of significant magic. This highlights a clear advantage of quantum simulation in such settings.

Additionally, based on this observation, even if the initial state possesses limited quantum resources, we may perform a truncation after the entanglement and magic of the state grow sufficiently, in accordance with Eq.~\eqref{equ:tri-ineq}. The error beyond this truncation can be treated as an average value, and only the error in the early stage of evolution requires rigorous analysis.

It is also worth noting that, when plotting FIG.~\ref{fig: threeDistribution} in the main text,
we selected the state at $t=0$ and $t=3.9$ in FIG.~\ref{fig: increase}(a) as $\ket{\psi_{LL}}$ and
$\ket{\psi_{HH}}$, respectively, and the state at $t=0.4$ in FIG.~\ref{fig: increase}(b) as $\ket{\psi_{LH}}$.

\begin{figure}[t]
    \centering
    \includegraphics[width=\linewidth]{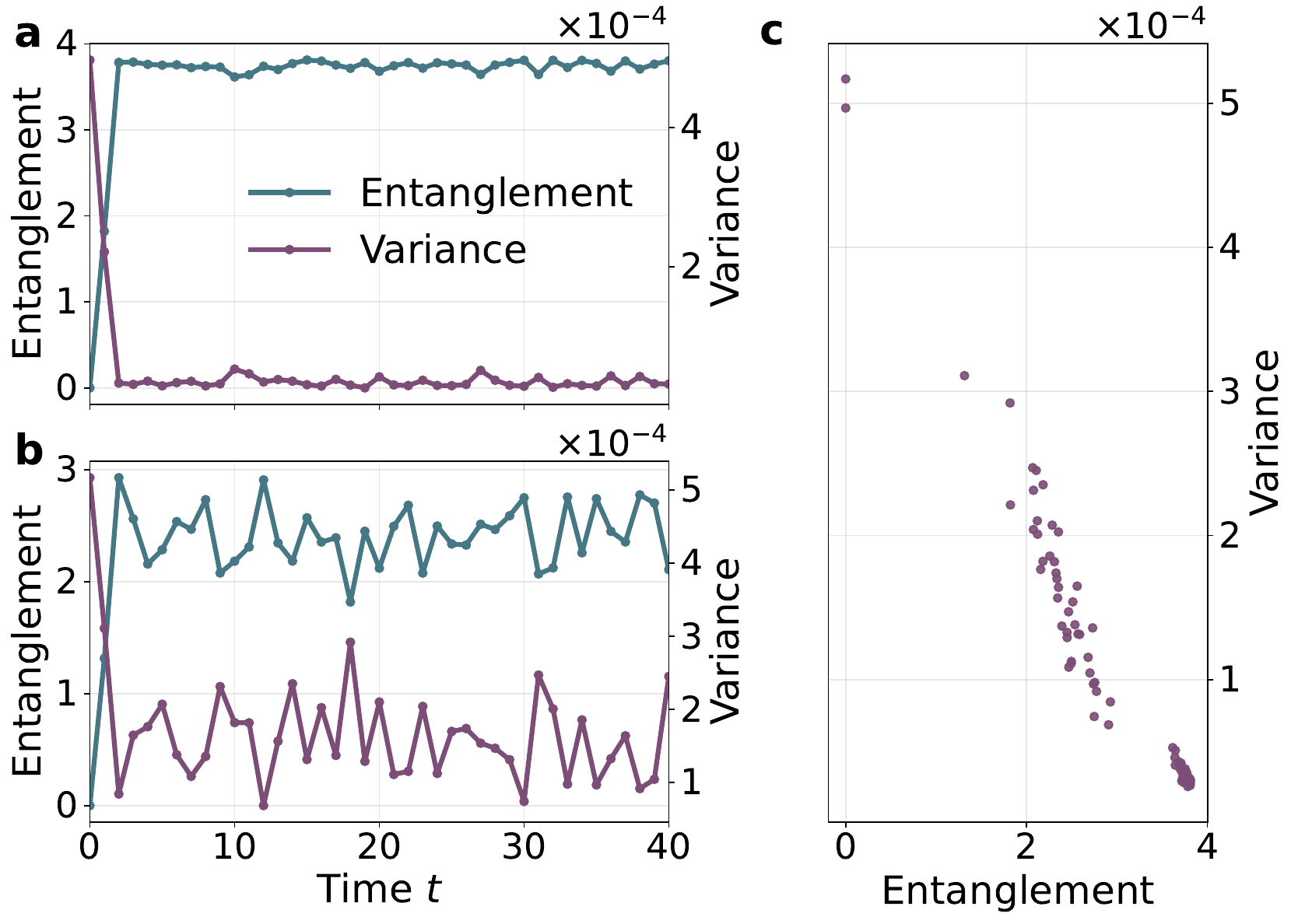}
    \caption{Entanglement (4-part entropy) reduces the variance of long-time simulation error. The long-time error, for $t=10$ here, is $\norm{\sU_p(\delta t)^r-\sU_0(\delta t)^r\ket{\psi}}$, where $p=2$, $r=100$, and $\delta t=0.1$. (a) Rapid entanglement growth reaching saturation as time increases; (b) Slow entanglement growth without reaching saturation. (c) displays the corresponding entanglement-variance correlations for the data shown in (a) and (b).}
    \label{fig:0.1longtime}
\end{figure}

\subsection{Long-time simulation}
To demonstrate how the entanglement of the initial state mitigates the variance of long-time simulation errors, we adopt the identical initial states as those presented in FIG.~\ref{fig: var_ent}, with the corresponding results illustrated in FIG.~\ref{fig:0.1longtime}. In this part, we use 1D Heisenberg model
\begin{equation}
    H=h\sum_{j=1}^NX_j+J\sum_{j=0}^{j-1}(X_jX_{j+1}+Y_jY_{j+1}+Z_jZ_{j+1}),
\end{equation}
where $N=10$, $J=1$ and $h=0.2$.
We utilize the second-order Trotter formula, $\sU_2(\delta t)$ to simulate the ideal evolution $\sU_0(\delta t)=\exp(-iH\delta t)$, where $\delta t=0.1$.

For each initial state $\ket{\psi_t}$, we generate another 2000 states $\ket{\psi_t^{(p)}}$ by acting local random unitary, simulating the ensemble $\cE(\psi_t)$. Then we calculate the simulation error $e^{(p)}_t=\norm{\left(\sU_2(\delta t)^r-\sU_0(\delta t)^r\right)\ket{\psi_k^{(p)}}}$ of each $\ket{\psi_t^{(p)}}$ and further evaluate the sample variance of these error values, which is then visualized in FIG.~\ref{fig:0.1longtime}. The numerical result indicates that when $\delta t$ is small, the long-time simulation error behaves similarly to the one-step one.

\clearpage
\onecolumngrid
\section*{Supplementary Material}

\section{Haar measure and unitary design}
In this part, we first review the basics of random unitary and design, which would be helpful for the proofs in the later sections.
We use $\mathcal{H}_d$ to represent the d-dimensional Hilbert space. $\mathcal L(\mathcal H_d)$ is the set of linear operators that act on $\mathcal H_d$. $D(\mathcal{H}_d)$ represents the set of density matrices on $\mathcal{H}_d$. We first define the moment operator.

\begin{definition}[$k$-th Moment operator]
The $k$-th moment operator, with respect to the probability measure $\nu$, is defined as $\mathcal R_\nu^{(k)}:\mathcal L(\mathcal H_d^{\otimes k})\rightarrow\mathcal L(\mathcal H_d^{\otimes k})$:
\begin{equation}
    \mathcal R_\nu^{(k)}(O)=\mathop{\bE}_{U\sim\nu}\left[U^{\otimes k}OU^{\dagger\otimes k}\right],
\end{equation}
for all operators in $\mathcal L(\mathcal H_d^{\otimes k})$.
\end{definition}

Let $S_k$ be the symmetric group of order $k$. For distinct $a_1, \cdots, a_l$, we use the notation $(a_1\cdots a_l)$ to denote a cyclic permutation in $S_k$ , which acts as $a_1\mapsto a_2, \cdots, a_{l-1}\mapsto a_l, a_l\mapsto a_1$. For $\sigma \in S_k$, define a unitary operator $W_\sigma$ acting on $H_d^{\otimes k}$ by
\begin{equation}
	W_\sigma=\sum_{i_1,\cdots,i_k=1}^d \ketbra{i_{\sigma^{-1}(1)}\cdots i_{\sigma^{-1}(k)}}{i_1\cdots i_k},
\end{equation}
where $\{\ket{i}\}_{i=1}^d$ is an orthonormal basis of $\mathcal{H}_d$. One can calculate the $k$-th moment operator using $W_\sigma$, as shown in Lemma \ref{lem: haar_random}.

\begin{lemma}[Theorem 10 in Ref.~\cite{Mele_2024}]\label{lem: haar_random}
	Let $O\in\mathcal L(\mathcal H_d^{\otimes k})$. The moment operator can then be expressed as a linear combination of permutation operators:
	\begin{equation}
		\mathop{\mathbb E}_{U\sim \mathrm{Haar}}\left[U^{\otimes k}OU^{\dagger\otimes k}\right]=\sum_{\sigma\in S_k}c_{\sigma}(O)W_\sigma,
	\end{equation}
    where the coefficients $c_\sigma(O)$ can be determined by solving the following linear system of $k!$ equations:
    \begin{equation}
        \Tr(W_\sigma^\dagger O)=\sum_{\sigma\in S_k}c_\sigma(O)\Tr(W_\sigma^\dagger W_\sigma).
    \end{equation}
    This system always has at least one solution.
\end{lemma}
A commonly used concept is the unitary $k$-design, defined as follows.
\begin{definition}[Unitary $k$-design]
	Let $\nu$ be a probability distribution over a set of unitaries $S\subseteq U(d)$. The distribution $\nu$ is unitary $k$-design if and only if:
	\begin{equation}
		\mathop{\mathbb E}_{V\sim\nu}\left[V^{\otimes k}OV^{\dagger \otimes k}\right]=\mathop{\mathbb E}_{U\sim\mathrm{Haar}}\left[U^{\otimes k}OU^{\dagger\otimes k}\right].
	\end{equation}
    for all $O\in\mathcal L((\mathcal H_d)^{\otimes k})$.
\end{definition}
For random Haar random states ($\ket{\psi}=U\ket{\psi_0},U\sim\mathrm{Haar}$), we have the $k$-th moment of states \cite{harrow2013church}.
\begin{lemma}\label{lem:haar_random}
	For Haar random pure state $\ket{\psi}\in\mathcal{H}_d$, we have
	\begin{equation}
		\mathop{\mathbb E}_{\psi\sim \mathrm{Haar}}\ketbra{\psi}{\psi}^{\otimes k}:=\mathop{\bE}_{U\sim\mathrm{Haar}}U^{\otimes k}\ketbra{0}{0}^{\otimes k}U^{\dagger\otimes k}=\frac{1}{(d+k-1)\cdots(d+1)d}\sum_{\sigma\in S_k}W_\sigma
	\end{equation}
\end{lemma}
Similarly we have the definition of state $k$-design.
\begin{definition}[state $k$-design]
	Let $\nu$ be a probability distribution over a set of states $S\subseteq \mathbb{C}^d$. The distribution $\nu$ is said to be a state $k$-design (or also spherical $k$-designs) if and only if:
	\begin{equation}
		\mathop{\mathbb E}_{\psi\sim\nu}\left[\ketbra{\psi}{\psi}^{\otimes k}\right]=\mathop{\mathbb E}_{\psi\sim\mathrm{Haar}}\left[\ketbra{\psi}{\psi}^{\otimes k}\right].
	\end{equation}
\end{definition}

\section{Variance and Entanglement}
We present a more specific and complete version of Theorem 1 in the main text. First we give one Lemmas.
\begin{fact}
    \label{lem: 1-design}
    $\cE(\psi_0)$ is an 1-design ensemble.
\end{fact} 
\begin{lemma}\label{lem: 2ndmoment}
    The second moment under $\cE(\psi_0)$ of a operator $A\in\mathcal L(\mathcal H_{2^N})$ with decomposition $A=\sum_{\{i\}}a_{\{i\}}P_{\{i\}}$ is
    \begin{equation}
        \mathop{\bE}_{\kp\sim\cE(\psi_0)}\bra{\psi}A\kp^2=\left(\frac{\Tr(A)}{d}\right)^2+\sum_{\{i\}\ne\{0\}}a_{\{i\}}^2\Tr(\ketbra{\psi_0}{\psi_0}^{\otimes 2}\bigotimes_{r=1}^{N}\left(\frac{2(1-\delta_{i_r0})}{3}\mathbb F+\frac{4\delta_{i_r0}-1}{3}\mathbb I\right)),
    \end{equation}
     where $\{i\}$ represents an index string $(i_1i_2\dots i_N)$, and each $i_k=0,1,2,3$ represents $P_{i_k}=I,X,Y,Z$ respectively. 
\end{lemma}
\begin{proof}
    The second moment of $A$ is 
    \begin{equation}
        \begin{aligned}\mathop{\bE}_{\kp\sim\cE(\psi_0)}\bra{\psi}A\kp^2&=\mathop{\bE}_{\kp\sim\cE(\psi_0)}\Tr(\ketbra{\psi}{\psi}^{\otimes 2}A^{\otimes 2})\\
        &=\mathop{\bE}_{U_1\sim\mathrm{Haar}}\cdots\mathop{\bE}_{U_N\sim\mathrm{Haar}}\Tr(\left(\bigotimes_{r=1}^NU_r\right)^{\otimes 2}\ketbra{\psi_0}{\psi_0}^{\otimes 2}\left(\bigotimes_{r=1}^NU_r^\dagger\right)^{\otimes 2}A^{\otimes 2})\\
        &=\Tr(\ketbra{\psi_0}{\psi_0}^{\otimes 2}\left(\mathop{\bE}_{U_1\sim\mathrm{Haar}}\cdots\mathop{\bE}_{U_N\sim\mathrm{Haar}}\left(\bigotimes_{r=1}^NU_r^\dagger\right)^{\otimes 2}A^{\otimes 2}\left(\bigotimes_{r=1}^NU_r\right)^{\otimes 2}\right))
        \end{aligned},
    \end{equation}
    so we only need to discuss the expectation term in the trace. From
    \begin{equation}
         A = \sum_{\{i\}}a_{\{i\}}P_{\{i\}}= \sum_{i_1,i_2,\dots,i_N}a_{i_1,i_2,\dots,i_N}P_{i_1}\otimes P_{i_2}\otimes \cdots P_{i_N} \label{equ: paulidecom}
    \end{equation}
    we have
    \begin{equation}
        A^{\otimes 2} = \sum_{\{i\}}\sum_{\{j\}}a_{\{i\}}a_{\{j\}}P_{i_1}\otimes P_{i_2} \otimes\cdots P_{i_N}\otimes P_{j_1}\otimes P_{j_2}\otimes\cdots P_{i_N}.
    \end{equation}
    The the expectation term becomes
    \begin{equation}
        \mathop{\bE}_{U_1\sim\mathrm{Haar}}\cdots\mathop{\bE}_{U_N\sim\mathrm{Haar}}\left(\bigotimes_{r=1}^NU_i^\dagger\right)^{\otimes 2}A^{\otimes 2}\left(\bigotimes_{r=1}^NU_i\right)^{\otimes 2}=\sum_{\{i\}}\sum_{\{j\}}a_{\{i\}}a_{\{j\}}\bigotimes_{r=1}^{N}\left(\mathop{\bE}_{U_r\sim \mathrm{Haar}}U_r^{\dagger\otimes 2}(P_{i_r}\otimes P_{j_r})U_r^{\otimes 2}\right).
    \end{equation}
    By Lemma \ref{lem: haar_random}, we have
    \begin{align}
        \mathop{\bE}_{U\sim\mathrm{Haar}}[(U^\dagger)^{\otimes 2}(P_i\otimes P_j)U^{\otimes 2}]=\frac{d\Tr(P_iP_j)-\Tr(P_i)\Tr(P_j)}{d(d^2-1)}\mathbb F + \frac{d\Tr(P_j)\Tr(P_i)-\Tr(P_iP_j)}{d(d^2-1)}\mathbb I\label{equ: 2design},
    \end{align}
    where $d=2$, $\mathbb F$ represents the two-qubit SWAP gate and $\mathbb I$ is the two-qubit identity. For two single-qubit Pauli operators $P_i,P_j$,
    \begin{equation}
        \Tr(P_i)=2\delta_{i0},\quad\Tr(P_j)=2\delta_{j0}, \quad\Tr(P_iP_j)=2\delta_{ij}.
    \end{equation}
    Then Eq.~\eqref{equ: 2design} becomes 
    \begin{align}
        \mathop{\bE}_{U\sim\mathrm{Haar}}[(U^\dagger)^{\otimes 2}(P_i\otimes P_j)U^{\otimes 2}]=\frac{2}{3}(\delta_{ij}-\delta_{i0}\delta_{j0})\mathbb F+\frac{1}{3}(4\delta_{i0}\delta_{j0}-\delta_{ij})\mathbb I.
    \end{align}
    Therefore, if $i\ne j$, then the second moment in Eq.~\eqref{equ: 2design} becomes $0$. According to this result, if we are interested in the second moment of $A$, the only terms we need to consider are $\sum_{\{i\}} a_{\{i\}}^2P_{\{i\}}^{\otimes 2}$. When $\{i\}=(00\dots0):=\{0\}$, we have $a_{00\dots0}=\Tr(A)/d$, then the moment becomes
    \begin{align}
        \mathop{\bE}_{U_1\sim\mathrm{Haar}}\cdots\mathop{\bE}_{U_N\sim\mathrm{Haar}}\left(\bigotimes_{r=1}^NU_r^\dagger\right)^{\otimes 2}A^{\otimes 2}\left(\bigotimes_{r=1}^NU_r\right)^{\otimes 2}=\left(\frac{\Tr(A)}{d}\right)^2I^{\otimes 2}_{2^N}+\sum_{\{i\}\ne \{0\}}a_{\{i\}}^2\bigotimes_{r=1}^{N}\left(\frac{2(1-\delta_{i_r0})}{3}\mathbb F+\frac{4\delta_{i_r0}-1}{3}\mathbb I\right).\label{equ: commoncase}
    \end{align}
    Therefore, the second moment of $A$ under $\cE(\psi_0)$ is
    \begin{align}
        \mathop{\bE}_{\kp\sim\cE(\psi_0)}\bra{\psi}A\kp^2&=\Tr(\ketbra{\psi_0}{\psi_0}\left(\mathop{\bE}_{U_1\sim\mathrm{Haar}}\cdots\mathop{\bE}_{U_N\sim\mathrm{Haar}}\left(\bigotimes_{r=1}^NU_r^\dagger\right)^{\otimes 2}A^{\otimes 2}\left(\bigotimes_{r=1}^NU_r\right)^{\otimes 2}\right))\\
        &=\left(\frac{\Tr(A)}{d}\right)^2+\sum_{\{i\}\ne\{0\}}a_{\{i\}}^2\Tr(\ketbra{\psi_0}{\psi_0}^{\otimes 2}\bigotimes_{r=1}^{N}\left(\frac{2(1-\delta_{i_r0})}{3}\mathbb F+\frac{4\delta_{i_r0}-1}{3}\mathbb I\right)).
    \end{align}
    \textbf{Remark}: Actually this Lemma only needs that each $U_i$ is selected from a unitary 2-design ensemble to meet Eq.~\eqref{equ: 2design}.
\end{proof}
Now we give the proof of Theorem 1 in the main text, restated as follows.
\begin{theorem}\label{theo:var-ent}
    For a given $n$-qubit state $\ket{\psi_0}$, we consider the main term of Trotter error $E=\sum E_j$ and the ensemble $\cE(\psi_0)=\{(\bigotimes_{r=1}^nU_i)\ket{\psi_0}\}$, where each $U_i$ is a random gate of one qubit chosen over a unitary 2-design, then the variance of $s_E(\psi)=\bra{\psi}E^\dagger E\ket{\psi}$ is bounded by the trace distance between the reduced density matrix of $\ket{\psi_0}$ on some subsystem indexed by $j,j'$ and the maximally mixed state $I/d_{jj'}$,
    \begin{align}
        \mathop{\mathrm{Var}}_{\ket{\psi}\sim\cE(\psi_0)}[s_E(\psi)] \leq \sum_{j\le j'}a_{jj'}\Tr\left|\rho_{jj'}-\frac{I}{d_{jj'}}\right|.
    \end{align}
    One can also use the entanglement entropy $S(\rho)=-\Tr[\rho\log(\rho)]$ to rewrite the bound as
    \begin{align}
        \mathop{\mathrm{Var}}_{\ket{\psi}\sim\cE(\psi_0)}[s_E(\psi)] \leq \sum_{j\le j'}a_{jj'}\sqrt{2\log(d_{jj'})-2S(\rho_{jj'})}.
    \end{align}
    which indicates that the upper bound of the variance of the squared Trotter error decreases with the increasing of entanglement of the initial state.
    The coefficient $a_{jj'}=2\Tr(\widetilde{E_{jj'}}\widetilde{E^\dagger E})/d$, where $\widetilde A:=\sum_{P\in\cP_N}\abs{\Tr(PA)}P/d$ and
    \begin{equation}
        E_{jj'}:=\frac{E_{j'}^\dagger E_j+E_j^\dagger E_{j'}-\Tr(E_{j'}^\dagger E_j+E_j^\dagger E_{j'})I/d}{1+\delta_{jj'}}.\label{equ: Ejj'}
    \end{equation}
    We define $\supp(jj')=\supp(E_{jj'})$, $\rho_{jj'}=\Tr_{[N]\setminus\supp(jj')}(\ketbra{\psi_0}{\psi_0})$ and $d_{jj'}=2^{\abs{\supp(jj')}}$.
\end{theorem}
\begin{proof}
    From Lemma \ref{lem: 1-design}, $\cE(\psi_0)$ is 1-design, so the expectation is \cite{zhao_hamiltonian_2022}
    \begin{equation}
        \mathop{\bE}_{\ket{\psi}\sim\cE(\psi_0)} s_E(\psi)=\|E\|_F^2.
    \end{equation} 
    We can decompose $E^\dagger E$ as 
    \begin{equation}
        E^\dagger E=\norm{E}_F^2I_{2^N}+\sum_{\{s\}}b_{\{s\}}P_{\{s\}},\label{equ: paulidec}    
    \end{equation}
    then from Lemma \ref{lem: 2ndmoment}, the variance is
    \begin{align}
         \mathop{\mathrm{Var}}_{\ket{\psi}\sim\cE(\psi_0)}[s_E(\psi)] &= \mathop {\bE}_{\ket{\psi}\sim\cE(\psi_0)} \bra{\psi}E^\dagger E\ket{\psi}^2-\left(\mathop {\bE}_{\ket{\psi}\sim\cE(\psi_0)} \bra{\psi}E^\dagger E\ket{\psi}\right)^2\notag\\
         &=\left(\frac{\Tr(E^\dagger E)}{d}\right)^2+\sum_{\{s\}}b_{\{s\}}^2\Tr(\ketbra{\psi_0}{\psi_0}^{\otimes 2}\bigotimes_{r=1}^{N}\left(\frac{2(1-\delta_{s_r0})}{3}\mathbb F+\frac{4\delta_{s_r0}-1}{3}\mathbb I\right))-\|E\|_F^4\notag\\ &=\sum_{\{s\}}b_{\{s\}}^2\Tr(\ketbra{\psi_0}{\psi_0}^{\otimes 2}\bigotimes_{r=1}^{N}\left(\frac{2(1-\delta_{s_r0})}{3}\mathbb F+\frac{4\delta_{s_r0}-1}{3}\mathbb I\right))\notag\\
         &=\sum_{\{s\}}b_{\{s\}}^2\Tr(\ketbra{\psi_0}{\psi_0}^{\otimes 2}\mathcal F_{\supp(\{s\})}\otimes I_{[N]\setminus\supp(\{s\})}^{\otimes 2}),\label{equ: varresult}
    \end{align}
    where $\supp(\{s\}=(s_1s_2\dots s_N)):=\{k|s_k\ne0\}$ and $\mathcal F_{\supp(\{s\})}:=\bigotimes_{m\in\supp(\{s\})}\left(\frac{2\mathbb F_{m}-\mathbb I_{m}}{3}\right)$. Even though we have the varince formula,
    $b_{\{s\}}$ in Eq.~\eqref{equ: paulidec} does not explicitly depends on the local error terms $E_j$.

    In the following, we aim to bound the variance in terms of the local error terms $E_j$ in $E$. To this end, we make a new decomposition of $E^\dagger E$ compared to Eq.~\eqref{equ: paulidec} as
    \begin{equation}
        E^\dagger E=\|E\|_F^2I+\sum_{j\le j'}E_{jj'}\label{equ: reorg},
    \end{equation}
    where $E_{jj'}$ is traceless and Hermitian, defined as Eq.~\eqref{equ: Ejj'} from the error terms $E_j$. 
    Each $E_{jj'}$ consists of several Pauli strings,
    \begin{equation}
        E_{jj'}=\sum_{\{i\}_{jj'}}\alpha_{\{i\}_{jj'}}P_{\{i\}_{jj'}}=\sum_{\{i\}_{jj'}}\alpha_{\{i\}_{jj'}}\hat P_{\{i\}_{jj'}}\otimes I_{[N]\setminus\supp(\{i\}_{jj'})},\quad\alpha_{\{i\}_{jj'}} \in\mathbb R. \label{equ: paulistring}
    \end{equation}
    Here $\{i\}_{jj'}=(i_1i_2\dots i_N)_{jj'}$ represents index strings occurring in $E_{jj'}$, and
    \begin{equation}
        \hat P_{\{i\}_{jj'}}:=\bigotimes_{k\in\supp(\{i\}_{jj'})}P_{i_k}=\frac{\Tr_{[N]\setminus\supp(\{i\}_{jj'})}(P_{\{i\}_{jj'}})}{2^{N-\left|\supp(\{i\}_{jj'}\right|}}
    \end{equation}
    the nontrivial Pauli part.
    Apparently one has $\supp(\{i\}_{jj'})\subseteq \supp(E_{jj'})$. 

    We can relate $b_{\{s\}}$ in Eq.~\eqref{equ: paulidec} with $\alpha_{\{i\}_{jj'}}$.
    Indeed, each $b_{\{s\}}$ is the summation of several $\alpha_{\{i\}_{jj'}}$ from different $E_{jj'}$, i.e. $\{i\}_{jj'}=\{i\}_{ll'}$ for some $j,j',l,l'$. As such, we can write
    \begin{equation}
        b_{\{s\}}=\sum_{j\le j'}\sum_{\{i\}_{jj'}}\alpha_{\{i\}_{jj'}}\mathscr{I}_{\{i\}_{jj'}=\{s\}},\label{equ: bs}
    \end{equation}
    where $\mathscr{I}_p=1/0$ iff $p$ is true/false to count all possible $\alpha_{\{i\}_{jj'}}$ in the summation. Substituting one $b_{\{s\}}$ in Eq.~\eqref{equ: varresult} by Eq.~\eqref{equ: bs}, one obtains
    \begin{equation}
        \begin{aligned}
            \mathop{\mathrm{Var}}_{\psi\sim\cE(\psi_0)}[s_E(\psi_0)]&=\sum_{\{s\}}\sum_{j\le j'}\sum_{\{i\}_{jj'}}b_{\{s\}}\alpha_{\{i\}_{jj'}}\mathscr{I}_{\{i\}_{jj'}=\{s\}}\Tr(\ketbra{\psi_0}{\psi_0}^{\otimes 2}\mathcal F_{\supp(\{s\})}\otimes I_{[N]\setminus\supp(\{s\})}^{\otimes 2})\\
            &=\sum_{j\le j'}\sum_{\{i\}_{jj'}}b_{\{i\}_{jj'}}\alpha_{\{i\}_{jj'}}\Tr(\ketbra{\psi_0}{\psi_0}^{\otimes 2}\mathcal F_{\supp(\{i\}_{jj'})}\otimes I_{[N]\setminus\supp(\{i\}_{jj'})}^{\otimes 2})\\
            &=\sum_{j\le j'}\sum_{\{i\}_{jj'}}b_{\{i\}_{jj'}}\alpha_{\{i\}_{jj'}}\Tr(\rho_{jj'}^{\otimes 2}\mathcal F_{\supp(\{i\}_{jj'})}\otimes I_{\supp(jj')\setminus\supp(\{i\}_{jj'}) }^{\otimes 2})
        \end{aligned}\notag
    \end{equation}
    where in the last equality we trace out the qubit outside $\supp(jj')$, and again note that $\supp(\{i\}_{jj'})\subseteq\supp(jj')$.
    For simplicity, we denote $\mathcal T_{\{i\}_{jj'}}:=\mathcal F_{\supp(\{i\}_{jj'})}\otimes I_{\supp(jj')\setminus\supp(\{i\}_{jj'})}$, and have
    \begin{equation}
        \begin{aligned}
            \mathop{\mathrm{Var}}_{\psi\sim\cE(\psi_0)}[s_E(\psi_0)]&=\sum_{j\le j'}\sum_{\{i\}_{jj'}}b_{\{i\}_{jj'}}\alpha_{\{i\}_{jj'}}\Tr(\rho_{jj'}^{\otimes 2}\mathcal T_{\{i\}_{jj'}})\\
            &=\sum_{j\le j'}\sum_{\{i\}_{jj'}}b_{\{i\}_{jj'}}\alpha_{\{i\}_{jj'}}\left(\Tr[\left(\rho_{jj'}^{\otimes 2}-\frac{I_{jj'}^{\otimes 2}}{d_{jj'}}\right)\mathcal T_{\{i\}_{jj'}}]+\cancel{\frac{1}{d_{jj'}}\Tr(\mathcal T_{\{i\}_{jj'}})}\right)\\
            &\le\sum_{j\le j'}\sum_{\{i\}_{jj'}}\abs{b_{\{i\}_{jj'}}\alpha_{\{i\}_{jj'}}}\abs{\Tr[\left(\rho_{jj'}^{\otimes 2}-\frac{I_{jj'}^{\otimes 2}}{d_{jj'}}\right)\mathcal T_{\{i\}_{jj'}}]}
        \end{aligned}
    \end{equation}
    where the last equality is because $\Tr(\mathcal T_{\{i\}_{jj'}})\propto\Tr(2\mathbb F-\mathbb I)=0$ and $\supp{\{i\}_{jj'}}\neq \emptyset$. In addition, by Hölder's inequality,
    \begin{equation}
        \abs{\Tr[\left(\rho_{jj'}^{\otimes 2}-\frac{I_{jj'}^{\otimes 2}}{d_{jj'}}\right)\mathcal T_{\{i\}_{jj'}}]}\le\norm{\mathcal T_{\{i\}_{jj'}}}\Tr\abs{\rho_{jj'}^{\otimes 2}-\frac{I_{jj'}^{\otimes 2}}{d_{jj'}}},
    \end{equation}
    where $\norm{\cdot}$ is the spectral norm. For two qubits, $\left \|(2\mathbb F-\mathbb I)/3\right\|=\norm{\mathbb I}=1$. As such,
    \begin{equation}
        \norm{\mathcal T_{\{i\}_{jj'}}}=\prod_{k\in\supp(\{i\}_{jj'})}\norm{\frac{2\mathbb F_k-\mathbb I_k}{3}}\times\prod_{l\in\supp(jj')\setminus\supp(\{i\}_{jj'})}\norm{\mathbb I_l}=1.
    \end{equation}
    The trace distance on the double-copy can be bounded by the triangle inequality as
    \begin{align}
        \Tr\left|\rho^{\otimes 2}_{jj'}-\frac{I_{jj'}^{\otimes 2}}{d_{jj'}^2}\right|&=\left\|\rho^{\otimes 2}_{jj'}-\rho_{jj'}\otimes \frac{I_{jj'}}{d_{jj'}}+\rho_{jj'}\otimes \frac{I_{jj'}}{d_{jj'}}-\frac{I_{jj'}^{\otimes 2}}{d_{jj'}^2}\right\|_1\notag\\
         &\leq\|\rho_{jj'}\|_1\left\|\rho_{jj'}-\frac{I_{jj'}}{d_{jj'}}\right\|_1+\left\|\frac{I_{jj'}}{d_{jj'}}\right\|_1\left\|\rho_{jj'}-\frac{I_{jj'}}{d_{jj'}}\right\|_1\quad\text{(triangle inequality)}\notag\\
         &=2\Tr\left|\rho_{jj'}-\frac{I_{jj'}}{d_{jj'}}\right|.\notag
    \end{align}
    
    Consequently, we have
    \begin{align}\label{equ: trdis}
        \mathop{\mathrm{Var}}_{\kp\sim\cE(\psi_0)}{[s_E(\psi)]}\le\sum_{j<j'}\left(\sum_{\{i\}_{jj'}}2\abs{b_{\{i\}_{jj'}}\alpha_{\{i\}_{jj'}}}\right)\Tr\abs{\rho_{jj'}-\frac{I_{jj'}}{d_{jj'}}}=\sum_{j<j'}a_{jj'}\Tr\abs{\rho_{jj'}-\frac{I_{jj'}}{d_{jj'}}}.
    \end{align}
    Here we denote the coefficient $\sum_{\{i\}_{jj'}}2\abs{b_{\{i\}_{jj'}}\alpha_{\{i\}_{jj'}}}$ as $a_{jj'}$. In the following, we relate $a_{jj'}$ to the error terms as stated in Theorem.
    \begin{equation}
        \begin{aligned}
            a_{jj'}&=2\sum_{\{i\}_{jj'}}\sum_{\{s\}}\abs{b_{\{s\}}}\abs{a_{{\{i\}}_{jj'}}}\mathscr{I}_{\{i\}_{jj'}=\{s\}}\\
            &=\frac{2}{d}\sum_{\{i\}_{jj'}}\sum_{\{s\}}\abs{b_{\{s\}}}\abs{a_{\{i\}_{jj'}}}\Tr(P_{\{i\}_{jj'}}P_{\{s\}})\\
            &=\frac{2}{d}\Tr[\left(\sum_{\{i\}_{jj'}}\abs{a_{\{i\}_{jj'}}}P_{\{i\}_{jj'}}\right)\left(\sum_{\{s\}}\abs{b_{\{s\}}}P_{\{s\}}\right)]\\
            &=\frac{2}{d}\Tr(\widetilde{E_{jj'}}\widetilde{E^\dagger E})
        \end{aligned}
    \end{equation}
    where $\widetilde A:=\sum_{P\in\cP_N}\abs{\Tr(AP)}P/d, \forall A\in\mathcal L(\mathcal H_d)$. 
    
    Moreover, the trace distance of $\rho_{jj'}$ and $I/d_{jj'}$ can be further bounded by the relative entropy as
    \begin{align}
        \Tr\left|\rho_{jj'}-\frac{I_{jj'}}{d_{jj'}}\right|\leq\sqrt{2S(\rho_{jj'}\|\frac{I_{jj'}}{d_{jj'}})}=\sqrt{2\log{d_{jj'}}-2S(\rho_{jj'})}\label{equ: relative entropy}.
    \end{align}
    From Eq.~\eqref{equ: trdis} and Eq.~\eqref{equ: relative entropy}, the bound given by the entropy is
    \begin{equation}
        \mathop{\mathrm{Var}}_{\ket{\psi}\sim\cE(\psi_0)}[s_E(\psi)]\leq \sum_{j\le j'}a_{jj'}\sqrt{2\log{d_{jj'}}-2S(\rho_{jj'})}.
    \end{equation}
\end{proof}
\section{Kurtosis and Magic}
To characterize the statistical properties of states with identical magic, we apply random Clifford gates to the initial state. Given that the random Clifford group forms a unitary 3-design, the first three moments of $\bra{\psi}E^\dagger E\ket{\psi}$ can be computed directly using Haar-random results for any given error term $E$, independent of the initial state $\psi_0$.Therefore, the fourth moment is of particular interest. We compute it as presented in Lemma \ref{lem:4moment}.
\begin{lemma}\label{lem:4moment}
    The 4th moment of $s_E(\psi)$ under $\cM(\psi_0)$ is 
    \begin{equation}
        \mathop{\bE}_{\ket{\psi}\in\cM(\psi_0)} s_E(\psi)^4= \frac{24d(B-A)+6(1-M)[(d^2+3d)A-4B]}{4d(d-1)(d+1)(d+2)(d+4)},\label{equ:4moment}
    \end{equation}
    where
    \begin{equation}
        \begin{aligned}
            A = \Tr[\Pi_4^{sym}(E^{\dagger} E)^{\otimes 4} Q],\quad
            B = \Tr[\Pi_4^{sym}(E^{\dagger} E)^{\otimes 4}],\quad
            Q = \frac{1}{d^2}\sum_{P\in\mathcal P_N}P^{\otimes 4}.
        \end{aligned}\label{equ: ABQ}
    \end{equation}
    Here $\Pi_4^{sym}$ is the projector onto the totally symmetric space $\mathrm{Sym}_4(\mathbb C^d)$ and $M$ is the magic of the state $\ket{\psi_0}$.
\end{lemma}
\begin{proof}
    The 4th-moment function of  states in $\cM(\psi_0)$ is \cite{zhu2016clifford}
\begin{equation}
    \mathop{\bE}_{\ket{\psi}\sim \cM(\psi_0)}  \ketbra{\psi}{\psi}^{\otimes 4}=\beta_+Q\Pi^{sym}_4+\beta_-(1-Q)\Pi^{sym}_4
\end{equation}
where $Q$ and $\Pi^{sym}_4$ are defined in Eq.~\eqref{equ: ABQ}, and
\begin{equation}
    \begin{aligned}
        \beta_+&=\frac{6}{(d+1)(d+2)}\Tr[Q\Pi_4^{sym}(\ketbra{\psi}{\psi}^{\otimes 4})]\\
        \beta_-&=\frac{24}{(d-1)(d+1)(d+2)(d+4)}\left(1-\Tr[Q\Pi_4^{sym}(\ketbra{\psi}{\psi}^{\otimes 4})]\right)
    \end{aligned}\label{equ: betas}.
\end{equation}
By definition, we have
\begin{equation}
    \Pi_4^{sym}\ketbra{\psi}{\psi}^{\otimes 4}=\ketbra{\psi}{\psi}^{\otimes 4},
\end{equation}
then the trace in Eq.~\ref{equ: betas} becomes
\begin{equation}
    \Tr[Q\Pi_4^{sym}(\ketbra{\psi}{\psi}^{\otimes 4})]=\Tr[Q(\ketbra{\psi}{\psi}^{\otimes 4})]=\frac{1}{d^2}\sum_{P\in\cP_N}|\bra{\psi}P\ket{\psi}|^4=\frac{1}{d}(1-M).
\end{equation}
Here $M=1-\frac{1}{d}\sum_{P\in\cP_N}|\bra{\psi}P\ket{\psi}|^4$ is the magic of state $\ket{\psi}$. Then we have
\begin{equation}
    \mathop{\mathbb E}_{\ket{\psi}\in\cM(\psi_0)}\ket{\psi}\bra{\psi}^{\otimes 4}=\frac{6(1-M)(d-1)(d+4)-24(d-1-M)d}{d(d-1)(d+1)(d+2)(d+4)}Q\Pi_4^{sym}+\frac{24(d-1-M)}{(d-1)(d+1)(d+2)(d+4)}\Pi_4^{sym}.
\end{equation}
Thus, the 4th-moment of $s_E$ obtains
    \begin{equation}
        \begin{aligned}
            \mathop{\bE}_{\ket{\psi}\in \cM(\psi_0)}s_E(\psi)^4&=\mathop{\bE}_{\ket{\psi}\in\cM(\psi_0)}\Tr\left[(E^\dagger E)^{\otimes 4}\ket{\psi}\bra{\psi}^{\otimes 4}\right]=\Tr\left[(E^\dagger E)^{\otimes 4}\mathop{\bE}_{\ket{\psi}\in\cM(\psi_0)}\ket{\psi}\bra{\psi}^{\otimes 4}\right]\\
            &=\frac{6(d-1)(d+4)(1-M)-24(d-1-M)}{d(d^2-1)(d+2)(d+4)}A+\frac{24(d-1-M)}{d(d^2-1)(d+2)(d+4)}B\\
            &=6\times\frac{(d^2+3d)(1-M)A-4dA+4dB-4B(1-M_{lin})}{d(d^2-1)(d+2)(d+4)}\\
            &=\frac{24d(B-A)+6(1-M_{lin})\left[(d^2+3d)A-4B\right]}{d(d^2-1)(d+2)(d+4)},
        \end{aligned}
    \end{equation}
    where $A,B$ is defined in Eq.~\eqref{equ: ABQ}.
\end{proof}
Through Lemma \ref{lem:4moment}, we find that the fourth moment of $s_E(\psi)$ exhibits a linear dependence on magic for a given set of states with identical magic $M$. In practice, the error term $E$ we typically observe takes a form similar to $\sum a_jQ_j$, where each $Q_j$ is a local Pauli operator and each $a_j$ is a constant coefficient. Now by analyzing the sign of the dominant term in the coefficients of Eq.\eqref{equ:4moment}, we establish a negative linear relationship between the fourth moment of $s_E(\psi)$ and $M$ when $n$ is large. See Lemma \ref{lem:coeff}. Before presenting Lemma \ref{lem:coeff}, we first introduce an auxiliary result, Fact \ref{fact: unchange}.

\begin{fact}\label{fact: unchange}
For any $O\in\mathcal D(\mathcal H_d)$, we have 
\begin{equation}
	\sum_{P\in \cP_n}\Tr(OPOP)=d\Tr(O)^2
\end{equation}
\end{fact}
\begin{proof}
We use the tensor network to show this result. We have
\begin{align}
    \sum_{P\in\cP_n}\Tr(OPOP)
    =\sum_{P\in\cP_n}\vcenter{\hbox{\Qcircuit @C=0.3cm @R=.3cm {
         &\qw & \qw & \qw & \qw &\qw  \\
         \qwx[-1] & \gate{O} & \gate{P} &\qw & \link{1}{-1} & \qw\qwx[-1] \\
         \qwx[1] & \gate{O} &\gate{P} &\qw & \link{-1}{-1}  & \qw\qwx[1]\\
         &\qw & \qw & \qw & \qw &\qw
     }}}
     ={\vcenter{\hbox{\Qcircuit@C=0.3cm @R=.3cm {
        &\qw & \qw & \qw & \qw &\qw \\
        \qwx[-1] & \gate{O} &\multigate{1}{\sum P^{\otimes 2}}& \qw&\link{1}{-1} &\qw\qwx[-1] \\
        \qwx[1] &\gate{O} &\ghost{\sum P^{\otimes 2}}& \qw&\link{-1}{-1} &\qw\qwx[1]\\
        & \qw &\qw &\qw &\qw &\qw
     }}}}.\label{equ: tensor}
\end{align}
In addition, we have $\frac{I+XX+YY+ZZ}{2}=\mathrm{SWAP}$ \cite{Zhou2023performanceanalysis}, then
\begin{equation}
    \frac{1}{d}\sum_{P\in\cP_n}P^{\otimes 2}=\bigotimes_{i=1}^N\left(\frac{1}{2}\sum_{P\in\cP_1}P^{\otimes 2}\right)=\mathrm{SWAP}^{\otimes N}
    =\vcenter{\hbox{\Qcircuit @C=.5cm @R= .5cm{
     & \qw &\link{1}{-1} & \qw\\
     & \qw &\link{-1}{-1} & \qw
    }}},
\end{equation}
then Eq.\eqref{equ: tensor} becomes
\begin{align}
    \sum_{P\in\cP_n}\Tr(OPOP)
    =d\times\vcenter{\hbox{\Qcircuit @C=0.3cm @R=.3cm {
         &\qw & \qw & \qw &\qw & \qw &\qw  \\
         \qwx[-1] & \gate{O} &\qw & \link{1}{-1} &\qw & \link{1}{-1} & \qw\qwx[-1] \\
         \qwx[1] & \gate{O} &\qw &\link{-1}{-1} &\qw & \link{-1}{-1}  & \qw\qwx[1]\\
         &\qw & \qw & \qw & \qw &\qw &\qw
     }}}
     =d\times\vcenter{\hbox{\Qcircuit @C=0.3cm @R=.3cm {
     &\qw &\qw \\
     \qwx[-1] &\gate{O} &\qw\qwx[-1]\\
     \qwx[1] &\gate{O} &\qw\qwx[1]\\
     &\qw &\qw
     }}}=d\Tr(O)^2.
\end{align}
\end{proof}
\begin{lemma} \label{lem:coeff}
    Here $E=e^{i\theta}\sum_{j=1}^{m} a_jQ_j,Q_j\in\mathcal P_N$, where $m=\mathrm{poly}(N)$ is the number of error terms, each coefficient $a_j$ is real and satisfies $|a_i|=\Theta(1)$. Then when $N$ is large
    \begin{equation}
        4B\leq(d^2+3d)A
    \end{equation}
    where A and B are defined in Lemma \ref{lem:4moment}. 
\end{lemma}
\begin{proof}
    Calculating B directly, we get
    \begin{equation}
        24B = 6\Tr[(E^\dagger E)^4]+8\Tr[(E^\dagger E)^3]\Tr(E^\dagger E)+3\Tr[(E^\dagger E)^2]^2+6\Tr[(E^\dagger E)^2]\Tr(E^\dagger E)^2+\Tr(E^\dagger E)^4\label{equ:B}
    \end{equation}
    and
    \begin{equation}
        \begin{aligned}
            24A = \frac{1}{d^2}&\sum_{P\in\mathcal P_N}(6\Tr[(E^\dagger EP)^4]+8\Tr[(E^\dagger EP)^3]\Tr(E^\dagger EP)+3\Tr[(E^\dagger EP)^2]^2\\
            &+6\Tr[(E^{\dagger} EP)^2]\Tr(E^\dagger EP)^2+\Tr(E^\dagger EP)^4)
        \end{aligned}.\label{equ:A}
    \end{equation}   
    \begin{enumerate}[(a)]
    \item Analyze the magnitude of different terms of $B$. Set $d=2^n$ is the dimension of the system.
        \begin{align}
            E^\dagger E&=\sum_{i=1}^{m} a_i^2 I+2\sum_{\substack{i\neq j\\Q_{ij}=Q_{ji}}} a_ia_jQ_iQ_j=\Theta(m)I+\sum\alpha_{ij}Q_{ij},N\rightarrow \infty\label{equ:EE}\\
            \Rightarrow &\Tr(E^\dagger E)=\Theta(m)d=\mathrm{poly}(N)2^{N},N\rightarrow\infty
        \end{align}
        Here $\alpha_{ij}:=2a_ia_j$ and $Q_{ij}:=Q_iQ_j$. Therefore, the coefficients in the second part are all constants, i.e. $|\alpha_{ij}|=\Theta(1)$ and there are $\mathcal O(m^2)$ terms in total. Then
        \begin{align}
            (E^\dagger E)^2&=\left((\sum a_i^2)^2+4\sum \alpha_{ij}^2\right)I+4(\sum a_i^2)^2\left(\sum\alpha_{ij}Q_{ij}\right)+8\sum_{\substack{ij\neq kn\\Q_{ijkn}=Q_{knij}}}\alpha_{ij}\alpha_{kn}Q_{ij}Q_{kn}\\
            &=\Theta(m^2)I+\sum\beta_jP_{\beta_j}+\sum\gamma_kP_{\gamma_k},n\rightarrow \infty
        \end{align}
        so
        \begin{align}
            \Tr[(E^\dagger E)^2]=\Theta(m^2)d=\mathrm{poly}(N) 2^N,N\rightarrow\infty
        \end{align}
        Here each $\beta_j=\Theta(N^2)$ and there are $\mathcal O(N^2)$ terms in total. Each $\gamma_k=\Theta(1)$ and there are $\mathcal O(N^4)$ terms in total. Similarly, one can find that
        \begin{align}
            \Tr[(E^\dagger E)^3]&=\text{poly}(N)2^N\\
            \Tr[(E^\dagger E)^4]&=\text{poly}(N)2^N
        \end{align}
       Therefore, we establish that the terms in Eq.\eqref{equ:B} exhibit asymptotic orders of $\text{poly}(N)d$, $\text{poly}(N)d^2$,$\Theta(N^4d^2)$, $\Theta(N^4d^3)$, $\Theta(n^4d^4)$ respectively. Consequently, the highest-order term $\Tr(E^\dagger E)^4$ dominates the asymptotic behavior.
\item Now we proof the highest-order term in $B$ is controlled by the third and fifth terms in $A$:$$4\Tr(E^\dagger E)^4\leq (1+\frac{3}{d})\left(\sum_{P\in\mathcal P_N}3\Tr[(E^\dagger EP)^2]^2+\Tr(E^\dagger EP)^4\right)$$
        First, we have
        \begin{equation}
            \sum_{P\in\mathcal P_N}\Tr(E^\dagger EP)^4=\Tr(E^\dagger E)^4+ \sum_{P\neq I}\Tr(E^\dagger EP)^4\geq\Tr(E^\dagger E)^4.\label{equ:1_part}
        \end{equation}
        From Fact \ref{fact: unchange} and Cauchy-Schwarz inequality, we have
        \begin{equation}
        \begin{aligned}\label{equ: 3_part}
            \sum_{P\in\mathcal P_N}\Tr(E^\dagger EPE^\dagger EP)^2&\ge\frac{1}{d^2}\left(\sum_{P\in\cP_n}\Tr(E^\dagger EPE^\dagger EP)\right)^2\\
            &=\frac{1}{d^2}\left(d\Tr(E^\dagger E)^2\right)^2=\Tr(E^\dagger E)^4.
        \end{aligned}
        \end{equation}
        Then from Eq.\eqref{equ:1_part} and Eq.\eqref{equ: 3_part}, we have
        \begin{equation}
            \begin{aligned}
                4\Tr(E^\dagger E)^4&=3\Tr(E^\dagger E)^4+\Tr(E^\dagger E)^4\\
                &\le 3\sum_{P\in\mathcal P_N}\Tr(E^\dagger EPE^\dagger EP)^2+\Tr(E^\dagger EP)^4\\
                &\le (1+\frac{3}{d})\left(3\sum_{P\in\mathcal P_N}\Tr(E^\dagger EPE^\dagger EP)^2+\Tr(E^\dagger EP)^4\right).
            \end{aligned}
        \end{equation}
\item Now we discuss the remain terms of $A$. Follow the same discussion in (a), we have
	\begin{equation}
		\left|\Tr[(E^\dagger EP)^k]\right|=\mathrm{poly}(N)d,
	\end{equation}
	where $k$ is a constant. Then the order of the summation of first terms of $A$ obtains
	\begin{equation}
		\left|\frac{1}{d^2}\sum_{P\in\cP_n}\Tr[(E^\dagger EP)^4]\right|\sim \frac{1}{d^2}\times d^2\times\mathrm{poly}(N)d=\mathrm{poly}(N)d.
	\end{equation}
	Both the second and fourth terms have a factor $\Tr(E^\dagger EP)$, which is zero unless $P\in\{Q_i\}_{i=1}^m$. Thus, there are actually only $m$ terms at most within each summation. Then we have
	\begin{equation}
		\begin{aligned}
			\left|\frac{1}{d^2}\sum_{P\in\cP_N}\Tr\left[(E^\dagger EP)^3\right]\Tr(E^\dagger EP)\right|&\sim \frac{1}{d^2}\times m\mathrm{poly}(N)d^2=\mathrm{poly}(N)\\
			\left|\frac{1}{d^2}\sum_{P\in\cP_n}\Tr\left[(E^\dagger EP)^2\right]\Tr(E^\dagger EP)^2\right|&\sim \frac{1}{d^2}\times m\mathrm{poly}(N)d^3=\mathrm{poly(N)}d.
		\end{aligned}
	\end{equation}
	Therefore, when $n$ is large enough, these three terms have negligible influence on the scale of $A$ compared to the third and fifth terms. 
    \end{enumerate}
	In conclusion, there always exists $N_0$ s.t.
        \begin{equation}
            4\Tr(E^\dagger E)^4\leq(d^2+3d)A\notag, N\geq N_0.
        \end{equation}
        As $\Tr(E^\dagger E)^4$ dominates the asymptotic behavior of $B$, we have
        \begin{equation}
            4B\leq (d^2+3d)A,
        \end{equation}
        when $N$ is large.

\end{proof}
As previously noted, the first three moments of $\bra{\psi}E^\dagger E\ket{\psi}$ are independent of $M$ and depend solely on the specific form of $E$. Therefore, the relationship between kurtosis and $M$ is reflected solely in the fourth moment, allowing us to derive Theorem \ref{theo:formerkurtosis}.
\begin{theorem}\label{theo:formerkurtosis}
    The set of all states with the same magic $M$ exhibits a decrease in the kurtosis of Trotter error $\bra{\psi}E^\dagger E\ket{\psi}$ as the magic increases. Specifically, the kurtosis satisfies:
    \begin{equation*}
        \mathop{\mathrm{Kur}}_{\ket{\psi}\sim\cM(\psi_0)}\left(\bra{\psi}E^\dagger E\ket{\psi}\right)=\alpha +\beta M\ge0
    \end{equation*}
    where $\alpha,\beta$ are some constants related to $E$, and $\beta<0$ when $n$ is large.
\end{theorem}
\begin{proof}
    In the following proof, we abbreviate $\cM(\psi_0)$ to $\cM$. As the Clifford group forms a unitary 3-design, we have
    \begin{equation}
        \mathop{\bE}_{\ket{\psi}\sim\cM}\ket{\psi}\bra{\psi}^{\otimes k}=\mathop{\bE}_{\ket{\psi}\sim\mathrm{Haar}}\ket{\psi}\bra{\psi}^{\otimes k},\quad\forall k\le3.
    \end{equation}
    Thus we have
    \begin{align}
        m_1&=\mathop{\bE}_{\ket{\psi}\sim\cM}\bra{\psi}E^\dagger E\ket{\psi}=\frac{\Tr(E^\dagger E)}{d}\\
        m_2&=\mathop{\bE}_{\ket{\psi}\sim\cM}\bra{\psi}E^\dagger E\ket{\psi}^2=\frac{\Tr(E^\dagger E)^2+\Tr[(E^\dagger E)^2]}{d(d+1)}\\
        m_3&=\mathop{\bE}_{\ket{\psi}\sim\cM}\bra{\psi}E^\dagger E\ket{\psi}^3=\frac{\Tr(E^\dagger E)^3+3\Tr(E^\dagger E)^2\Tr(E^\dagger E)+2\Tr[(E^\dagger E)^3]}{d(d+1)(d+2)}.
    \end{align}
    So we have the variance
    \begin{equation}
            \mathop{\mathrm{Var}}_{\ket{\psi}\sim\cM}\left(\bra{\psi}E^\dagger E\ket{\psi}\right)=m_2-m_1^2
    \end{equation}
    and the fourth central moment
    \begin{equation}
        \begin{aligned}
            &\mathop{\bE}_{\ket{\psi}\sim\cM}\left(\bra{\psi}E^\dagger E\ket{\psi}-\mathop{\bE}_{\ket{\psi}\sim\cM}\bra{\psi}E^\dagger E\ket{\psi}\right)^4\\
            =&\mathop{\bE}_{\ket{\psi}\sim\cM}\left(\bra{\psi}E^\dagger E\ket{\psi}^4-4\bra{\psi}E^\dagger E\ket{\psi}^3m_1+6\bra{\psi}E^\dagger E\ket{\psi}^2m_1^2-4\bra{\psi}E^\dagger E\ket{\psi}m_1^3+m_1^4\right)\\
            =&m_4-4m_3m_1+6m_2m_1^2-3m_1^4,
        \end{aligned}
    \end{equation}
    where $m_4=\mathop{\bE}_{\ket{\psi}\sim\cM}\bra{\psi}E^\dagger E\ket{\psi}^4$ is given by Lemma~\ref{lem:4moment}. Then the kurtosis is 
    \begin{equation}
        \begin{aligned}
            0\le\mathop{\mathrm{Kur}}_{\ket{\psi}\sim\cM}(\bra{\psi}E^\dagger E\ket{\psi})&=\frac{\mathop{\bE}_{\ket{\psi}\sim\cM}\left(\bra{\psi}E^\dagger E\ket{\psi}-\mathop{\bE}_{\ket{\psi}\sim\cM}|\bra{\psi}E^\dagger E\ket{\psi}|\right)^4}{\left[\mathop{\mathrm{Var}}_{\ket{\psi}\sim\cM}\left(\bra{\psi}E^\dagger E\ket{\psi}\right)\right]^2}\\
            &=\frac{m_4-4m_3m_1+6m_2m_1^2-3m_1^4}{(m_2-m_1^2)^2}\\
            &=\alpha+\frac{4B-(d^2+3d)A}{4d(d^2-1)(d+2)(d+4)(m_2-m_1^2)^2}M\\
            &=\alpha+\beta M.
        \end{aligned}
    \end{equation}
    From Lemma~\ref{lem:coeff}, we know that $\beta<0$ when $n$ is large.
\end{proof}
\section{Long-time simulation}
\subsection{Long-time magic--kurtosis relation}
Lemma~\ref{lem:4moment} shows that the fourth-moment, and hence the kurtosis quantity considered in the main text, depends linearly on the magic of the input state for an arbitrary error operator \(E\). Importantly, the proof of Lemma~\ref{lem:4moment} does not require \(E\) to be a leading-order Trotter error, a one-step error, or to satisfy any locality or perturbative assumption. Therefore, the same linear relation remains valid if we replace \(E\) by the accumulated long-time error operator
\begin{equation}
    E_{\mathrm{LT}} = \sU_p(\delta t)^r-\sU_0(\delta t)^r,
\end{equation}
where \(\sU_0(\delta t)^r\) is the exact time-evolution operator and \(\sU_p(\delta t)^r\) is the product-formula approximation with \(r\) Trotter steps. In particular, for the long-time error distribution, the kurtosis can still be written in the form
\begin{equation}
    \mathop{\mathrm{Kur}}_{\ket{\psi}\sim\cM(\psi_0)}\left(\bra{\psi}E_{\mathrm{LT}}^\dagger E_{\mathrm{LT}}\ket{\psi}\right)
    =
    \alpha_{\mathrm{LT}} + \beta_{\mathrm{LT}} M(\psi_0),
\end{equation}
where $M(\psi_0)$ denotes the magic measure used in the main text, and \(\alpha_{\mathrm{LT}}\) and \(\beta_{\mathrm{LT}}\) depend on the long-time error operator \(E_{\mathrm{LT}}\).

\begin{figure}[h]
    \centering
    \includegraphics[width=0.5\linewidth]{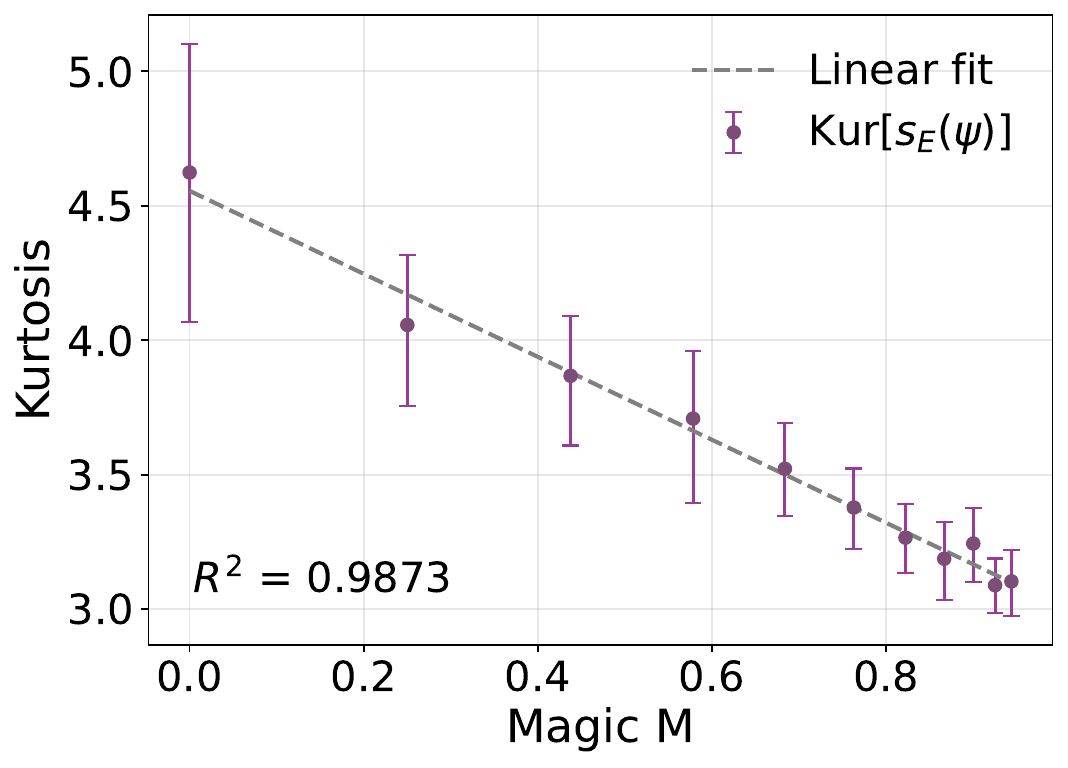}
    \caption{Long-time magic--kurtosis relation for the accumulated Trotter error. The simulation uses step size
    \(\delta t=0.1\) and \(r=100\) Trotter steps. Each point corresponds to
    the kurtosis of the long-time error distribution for states with a given
    magic value, and the dashed line denotes a linear fit. Consistent with the
    formal linear relation implied by Lemma~\ref{lem:4moment}, the data are well described by
    a linear function of magic. The fitted slope is negative for the system
    and parameters considered here, although the sign of this long-time
    coefficient is not fixed by our current analytical proof.}
    \label{fig:longtimekur}
\end{figure}

We emphasize, however, that Lemma~\ref{lem:4moment} by itself only establishes the linear dependence, not the sign of the coefficient \(\beta_{\mathrm{LT}}\). In Theorem~\ref{theo:formerkurtosis}, the negativity of the corresponding coefficient for the leading-order Trotter error was obtained using magnitude estimates specific to the perturbative, one-step error operator. These estimates no longer apply once \(E\) is replaced by the accumulated long-time error \(E_{\mathrm{LT}}\), whose operator content is generated by the full dynamics. Thus, our present analytical argument does not determine whether \(\beta_{\mathrm{LT}}\) is positive or negative in the long-time regime.

We therefore examine this question numerically. As shown in Fig.~\ref{fig:longtimekur}, for step size \(\delta t=0.1\) and \(r=100\) Trotter steps, the long-time magic--kurtosis relation remains approximately linear, and the fitted slope is negative (The accumulated error is generated using the same QIMF Hamiltonian as in the main text). This indicates that, for the physical systems and parameters studied here, higher magic continues to correlate with a smaller kurtosis of the long-time Trotter-error distribution. The observed negative slope should be viewed as numerical evidence rather than a rigorous consequence of our current proof.
\subsection{Long-time entanglement-variance relation}
For a long-time Trotter error, we can bound its error using the triangle inequality of norms.
\begin{fact}
    For a long-time simulation task, we cut $t$ into $r$ small pieces, then the error satisfies
    \begin{equation}
        \norm{(\sU_p(\delta t)^r-\sU_0(t))\ket{\psi_0}}
        \le\sum_{k=0}^{r-1}\norm{(\sU_p(\delta t)-\sU_0(\delta t))\sU_p(\delta t)^k\ket{\psi_0}}.
    \end{equation}
\end{fact}
\begin{proof}
    We derive this lemma by triangle inequality
    \begin{equation}
        \begin{aligned}
            \norm{(\sU_p(\delta t)^r-\sU_0(t))\ket{\psi_0}}
            &=\norm{\left(\sU_p(\delta t)^r-\sU_0(\delta t)\sU_p(\delta t)^{r-1}+\sU_0(\delta t)\sU_p(\delta t)^{r-1}-\sU_0(\delta t)^r\right)\ket{\psi_0}}\\
            &\le \norm{(\sU_p(\delta t)-\sU_0(\delta t))\sU_p(\delta t)^{r-1}\ket{\psi_0}}+\norm{(\sU_p(\delta t)^{r-1}-\sU_0(\delta t)^{r-1})\ket{\psi_0}}\\
            &\le....\\
            &\le\sum_{k=0}^{r-1}\norm{(\sU_p(\delta t)-\sU_0(\delta t))\sU_p(\delta t)^k\ket{\psi_0}}
        \end{aligned}
    \end{equation}
\end{proof}
For simplicity, we denote $\sum_{k=0}^{r-1}\norm{(\sU_p(\delta t)-\sU_0(\delta t))\sU_p(\delta t)^k\ket{\psi_0}}$ as $e_r$. When $\delta t$ is small enough, we can ignore the higher-order terms in $\sU_p(\delta t)-\sU_0(\delta t)$ and have
\begin{equation}
    e_r\approx\sum_{k=0}^{r-1}\norm{E\sU_p(\delta t)^k\ket{\psi}}\delta t^{p+1}=\delta t^{p+1}\sum_{k=0}^{r-1}\sqrt{s_E(\sU_p(\delta t)^k\ket{\psi})}.
\end{equation}
And we can prove that the variance of $e_r$ has the following inequality.
\begin{lemma}\label{aplem:sum_var}
The variance of $e_r$ satisfies
    \begin{equation}
        \mathop{\operatorname{Var}}_{\psi\sim\cE(\psi_0)}(e_r)\le r\delta t^{2p+2}\sum_{k=0}^{r-1}\mathop{\operatorname{Var}}_{\psi\sim\cE(\psi_0)}\left[\sqrt{s_E(\sU_p(\delta t)^k\ket{\psi})}\right]
    \end{equation}
\end{lemma}
\begin{proof}
    \begin{align}
    \mathop{\operatorname{Var}}_{\psi\sim\cE(\psi_0)}(e_r)=&\delta t^{2p+2}\sum_{k=0}^{r-1}\mathop{\operatorname{Var}}_{\psi\sim\cE(\psi_0)}\left[\sqrt{s_E(\sU_p(\delta t)^k\ket{\psi}}\right]+\delta t^{2p+2}\sum_{k\ne l}\mathop{\operatorname{Cov}}_{\psi\sim\cE(\psi_0)}\left[\sqrt{s_E(\sU_p(\delta t)^k\ket{\psi}},\sqrt{s_E(\sU_p(\delta t)^l\ket{\psi}}\right]\notag\\
    \le&\delta t^{2p+2}\sum_{k=0}^{r-1}\mathop{\operatorname{Var}}_{\psi\sim\cE(\psi_0)}\left[\sqrt{s_E(\sU_p(\delta t)^k\ket{\psi}}\right]\\
    &+\delta t^{2p+2}\sum_{k\ne l}\sqrt{\mathop{\operatorname{Var}}_{\psi\sim\cE(\psi_0)}\left[\sqrt{s_E(\sU_p(\delta t)^k\ket{\psi}}\right]\mathop{\operatorname{Var}}_{\psi\sim\cE(\psi_0)}\left[\sqrt{s_E(\sU_p(\delta t)^l\ket{\psi}}\right]}\notag\\
    \le&\frac{1}{2}\delta t^{2p+2}\sum_{k,l}\left(\mathop{\operatorname{Var}}_{\psi\sim\cE(\psi_0)}\left[\sqrt{s_E(\sU_p(\delta t)^k\ket{\psi}}\right]+\mathop{\operatorname{Var}}_{\psi\sim\cE(\psi_0)}\left[\sqrt{s_E(\sU_p(\delta t)^l\ket{\psi}}\right]\right)\\
    =&r\delta t^{2p+2}\sum_{k=0}^{r-1}\mathop{\operatorname{Var}}_{\psi\sim\cE(\psi_0)}\left[\sqrt{s_E(\sU_p(\delta t)^k\ket{\psi})}\right]
    \end{align}
\end{proof}
For each term in the summation, we have the following lemma.
\begin{lemma}\label{aplem:sqrt_X_var}
Given a bounded positive random variable $X$, the variance of $\sqrt{X}$ is bouned by
    \begin{equation}
        \operatorname{Var}(\sqrt{X})\le\frac{1}{\sqrt 2}\sqrt{\operatorname{Var}(X)}.
    \end{equation}
\end{lemma}
\begin{proof}
Assume that there is an another random variable $Y$ i.d.d. with $X$, then
    \begin{align}
        \operatorname{Var}(\sqrt{X})=&\frac{1}{2}\bE\left((\sqrt{X}-\sqrt Y)^2\right)
        \le \frac{1}{2}\bE\left(\abs{X-Y}\right)
        \le \frac{1}{2}\sqrt{\bE\left((X-Y)^2\right)}=\frac{1}{\sqrt 2}\sqrt{\operatorname{Var}(X)}.
    \end{align}
\end{proof}
By combining Theorem \ref{theo:var-ent} with the two preceding lemma, we arrive at Theorem 3 in the main text. We define $S_k=\{s\subseteq[N]|s=\operatorname{supp}(\sU_p^{\dagger k}E_j^\dagger E_{j'}\sU_p^k),\forall j,j'\}$ as the set of supports for each term in the expansion of operator $\sU_p^{\dagger k}E^\dagger E\sU_p^k$. Subsequently, we establish the following theorem.
\begin{theorem}
    Given a state $\ket{\psi_0}$ with entanglement structure $S(\rho_A)$, neglecting the influence of higher-order terms in $\delta t$, the variance of the $p$-th order Trotter error under $\cE(\psi_0)$ satisfies
    \begin{equation}
        \mathop{\operatorname{Var}}_{\kp\sim\cE(\psi_0)}[e_r]\le r\delta t^{2p+2}\sum_{k=0}^{r-1}\sqrt{\sum_{s\in S_k}c_{s}\sqrt{\log d_{s}-S(\rho_{s})}},
    \end{equation}
    where $c_s$ denotes a constant that is independent of $\psi_0$.
\end{theorem}
\begin{proof}
    From Theorem \ref{theo:var-ent},
    \begin{equation}
    \mathop{\operatorname{Var}}_{\psi\sim\cE(\psi_0)}\left[{s_E(\sU_p(\delta t)^k\ket{\psi})}\right]=\mathop{\operatorname{Var}}_{\psi\sim\cE(\psi_0)}\left[{s_{E\sU_p(\delta t)^k}(\ket{\psi})}\right]\le\sum_{j\le j'}a^{(k)}_{jj'}\sqrt{2\log{d^{(k)}_{jj'}}-2S(\rho^{(k)}_{jj'})}.
    \end{equation}
    It should be noted that the constants $a_{jj'}^{(k)}$ and the support of $\rho_{jj'}^{(k)}$ here are different from those in Theorem \ref{theo:var-ent}; they are obtained by replacing the error matrix $E$ with $E\sU_p(\delta t)^k$ respectively. The support of $\sU_p(\delta t)^{\dagger k}E^\dagger E\sU_p(\delta t)^k$ linearly grows with $k$ as $\sU_p(\delta t)$ has the brickwork architecture consisting of several local unitary. Therefore, the results still exhibits the property that the entanglement of $\psi_0$ reduces the variance. Combining Lemma \ref{aplem:sqrt_X_var},
    \begin{equation}
        \mathop{\operatorname{Var}}_{\psi\sim\cE(\psi_0)}\left[\sqrt{s_E(\sU_p(\delta t)^k\ket{\psi})}\right]\le\frac{1}{\sqrt 2}\sqrt{\sum_{j\le j'}a^{(k)}_{jj'}\sqrt{2\log{d^{(k)}_{jj'}}-2S(\rho^{(k)}_{jj'})}},
    \end{equation}
    and Lemma \ref{aplem:sum_var},
    \begin{align}
        \mathop{\operatorname{Var}}_{\kp\sim\cE(\psi_0)}[e_r]\le &r\delta t^{2p+2}\sum_{k=0}^{r-1}\frac{1}{\sqrt 2}\sqrt{\sum_{j\le j'}a^{(k)}_{jj'}\sqrt{2\log{d^{(k)}_{jj'}}-2S(\rho^{(k)}_{jj'})}}\\
        =&r\delta t^{2p+2}\sum_{k=0}^{r-1}\sqrt{\sum_{s\in S_k}c_s\sqrt{\log d_{s}-S(\rho_{s})}},
    \end{align}
    where the summation $\sum_{A_k}$ sums over the supports of $\sU_p(\delta t)^{\dagger k}E_j^\dagger E_{j'}\sU_p(\delta t)^k$s.
\end{proof}

\clearpage
\end{document}